%% file: mainLong.tex
\documentclass[a4paper,UKenglish,cleveref, autoref]{lipics-v2019}
\nolinenumbers

\usepackage{amsmath,amssymb}
\usepackage{wasysym}
\usepackage{environ}
\usepackage{xargs}
\usepackage{import}
\usepackage{url}
\usepackage{graphics}
\usepackage{color}
\usepackage{cmll}

\usepackage{tikz}
\usetikzlibrary{matrix,arrows,calc}
\tikzset{
  arrowstyle/.style={
    double,
    ->,
    shorten <=3pt,
    shorten >=3pt}}
\usetikzlibrary{decorations.markings}
\usetikzlibrary{decorations.pathmorphing, backgrounds,
  positioning, fit, shapes.geometric, shapes.multipart,
  decorations.pathreplacing}

\usepackage{wrapfig}
\usepackage{pnLuc}
\usepackage{multirow}
\usepackage{xspace}
\usepackage{forest}
\usepackage{mathtools}

\title{Glueability of resource proof-structures: inverting the Taylor
  expansion (long version)} 

\titlerunning{Glueability of resource proof-structures: inverting the Taylor
	expansion}

\author{Giulio Guerrieri}{University of Bath, Department of Computer Science, Bath,
  UK 
}{g.guerrieri@bath.ac.uk}{0000-0002-0469-4279}{}

\author{Luc Pellissier}{Université de Paris, IRIF, CNRS, F-75013 Paris, France
}{pellissier@irif.fr}{}{}

\author{Lorenzo {Tortora de Falco}}{Università Roma Tre, Dipartimento di Matematica e Fisica, Rome, Italy
}{tortora@uniroma3.it}{}{}

\authorrunning{G. Guerrieri and L. Pellissier and L. Tortora de Falco}

\Copyright{Giulio Guerrieri and Luc Pellissier and Lorenzo {Tortora de Falco}}



\category{}

\relatedversion{}

\supplement{}

\EventEditors{John Q. Open and Joan R. Access}
\EventNoEds{2}
\EventLongTitle{42nd Conference on Very Important Topics (CVIT 2016)}
\EventShortTitle{CVIT 2016}
\EventAcronym{CVIT}
\EventYear{2016}
\EventDate{December 24--27, 2016}
\EventLocation{Little Whinging, United Kingdom}
\EventLogo{}
\SeriesVolume{42}
\ArticleNo{23}

\input{macros}
\newcommand{\tikzsetnextfilename}[1]{}

\newcommand{\version}{1}

\newcommand{\SLV}[2]{\ifthenelse{\equal{\version}{0}}{#1}{#2}}

\begin{document}

\maketitle

\begin{abstract}
  A Multiplicative-Exponential Linear Logic (MELL) proof-structure can be expanded into a set of resource proof-structures: its Taylor expansion.
  We introduce a new criterion characterizing those sets of resource proof-structures that are part of the Taylor expansion of some MELL
  proof-structure, through a
  rewriting system acting both on resource and MELL proof-structures.
  
\end{abstract}

\input{introductionNew}
\input{proof-nets}
\input{taylor}
\input{rules}
\input{unwinding}
\input{naturality}

\input{glueability}
\input{general-case}
\input{conclusion}

\phantomsection
\bibliography{Biblio.bib}
\addcontentsline{toc}{section}{References}

\newpage
\appendix
\part*{Technical Appendix}
\addcontentsline{toc}{section}{Technical Appendix}
\input{naturality-appendix}
\input{not-eta-expansed}

\end{document}

%% file: macros.tex
\newcommand{\ie}{\textit{i.e.}\xspace}
\newcommand{\eg}{\textit{e.g.}\xspace}
\newcommand{\polyadic}{$\DiLL_0$\xspace}
\newcommand{\polyadicdaimon}{$\DiLL_0^\maltese$\xspace}
\newcommand{\filled}{filled\xspace}

\newcommand{\ReflexiveTransitive}{\circlearrowleft}

\newcommand{\TreeT}{\ensuremath{\mathcal{F}}}

\newcommand{\CatFont}[1]{\mathbf{#1}}

\newcommand{\Rel}{\ensuremath{\CatFont{Rel}}} 

\newcommand{\PolyPN}{\ensuremath{\CatFont{qDiLL_0^\maltese}}}
\newcommand{\PPolyPN}{\ensuremath{\mathfrak{P}\PolyPN}}

\newcommand{\Path}{\ensuremath{\CatFont{Path}}}

\newcommand{\qMELL}{\ensuremath{\CatFont{qMELL}^\maltese}}

\newcommand{\SystemFont}[1]{\mathsf{#1}}
\newcommand{\LL}{\ensuremath{\SystemFont{LL}}\xspace}
\newcommand{\MELL}{\ensuremath{\SystemFont{ME}\LL}\xspace}
\newcommand{\MELLdaimon}{\ensuremath{\SystemFont{ME}\LL^\maltese}\xspace}

\newcommand{\DiLL}{\ensuremath{\SystemFont{Di}\LL}\xspace}
\newcommand{\Ax}{\mathbin{\texttt{ax}}}
\newcommand{\Cut}{\mathbin{\texttt{cut}}}
\newcommand{\Ex}{\mathbin{\texttt{exc}}}
\newcommand{\One}{\mathbin{\mathbf{1}}}
\newcommand{\BoxR}{\mathbin{\mathsf{Box}}}
\newcommand{\Mix}{\mathbin{\texttt{mix}}}
\newcommand{\Maybe}[1]{#1}
\newcommand{\MMaybe}[1]{{}_{\wn}[#1]}
\newcommand{\emptynet}{\varepsilon}

\newcommand{\PowerSet}[1]{\mathfrak{P}(#1)}
\newcommand{\op}[1]{\ensuremath{{#1}^{\mathrm{op}}}}

\newcommand{\ax}{\texttt{ax}}
\newcommand{\cut}{\texttt{cut}}
\newcommand{\FlagType}{\ensuremath{\mathsf{c}}}
\newcommand{\VertType}{\ensuremath{\mathsf{\ell}}}

\newcommand{\BoxFunction}{\ensuremath{\mathsf{box}}}

\newcommand{\TaylorSym}{\ensuremath{\mathcal{T}}}
\newcommand{\Taylor}[1]{\ensuremath{\TaylorSym(#1)}}\newcommand{\ProtoTaylor}[1]{\ensuremath{\TaylorSym^{\mathsf{proto}}(#1)}}
\newcommand{\FatTaylor}[1]{\ensuremath{\TaylorSym^{\maltese}(#1)}}
\newcommand{\FatTaylorEta}[1]{\ensuremath{\TaylorSym^{\maltese}_{\eta}(#1)}}
\newcommand{\NatTransf}{\ensuremath{\mathfrak{T}^{\maltese}}}
\newcommand{\NatTransfEta}{\ensuremath{{\mathfrak{T}^{\maltese}_{\eta}}}}
\newcommand{\Projection}{\NatTransf}

\newcommand{\Nat}{\ensuremath{\mathbb{N}}}

\newcommand{\mathnodes}[4]{
  \matrix (m) [
  matrix of math nodes,
  row sep=#1em,
  column sep=#2em,
  minimum width=#3em,
  ampersand replacement = \&]
  {#4}
}

\NewEnviron{commutativediagram}[4]{
  \begin{center}
    \begin{tikzpicture}[>=stealth]
      \mathnodes{#1}{#2}{#3} { #4 } \BODY
    \end{tikzpicture}
  \end{center}

}

\newcommand{\To}{\Rightarrow}

\tikzset{
  mapsto/.style={
    |->,
    decoration={
      markings,
      mark=at position 1 with {\arrow[scale=1.7]{>}}
    },
    postaction={decorate},
    shorten >=0.4pt}
}

\newcommand\pnrewrite[1][\quad]{%
  \mathrel{%
    \begin{tikzpicture}[%
      baseline={(current bounding box.south)}
      ]
      \node[%
      ,inner sep=.44ex
      ,align=center
      ] (tmp) {$\scriptstyle #1$\vphantom{\quad}};
      \path[%
      ,draw,<-
      ,decorate,decoration={%
        ,snake
        ,amplitude=0.7pt
        ,segment length=1.2mm,pre length=3.5pt
      }
      ] 
      (tmp.south east) -- (tmp.south west);
    \end{tikzpicture}
  }
}

\newcommand{\smallproofnets}{0.75}
\newcommand{\beforepn}{-1cm}
\newcommand{\magicnumber}{\smallproofnets}

\usepackage{trimclip}
\DeclareMathOperator{\WHY}{\textup{?}}
\def\?{\setbox0=\hbox{$\WHY$}\raisebox{.2\ht0}{\clipbox{0pt .2\ht0 0pt
      -.1\ht0}{?}}}

\newcommand{\Contr}
{\textup{\ooalign{\raisebox{0ex}{$\?$}\cr\raisebox{-.5ex}{$\scriptscriptstyle
        c$}}}}
\newcommand{\Der}
{\textup{\ooalign{\raisebox{0ex}{$\?$}\cr\raisebox{-.5ex}{$\scriptscriptstyle
      d$}}}}
\newcommand{\Weak}
{\textup{\ooalign{\raisebox{0ex}{$\?$}\cr\raisebox{-.5ex}{$\scriptscriptstyle
      w$}}}}

%% file: introductionNew.tex
\section{Introduction}
\label{sec:intro}

\subparagraph{Resource \texorpdfstring{$\lambda$}{lambda}-calculus and the Taylor expansion}

Girard's linear logic (\LL, \cite{Girard:1987}) is a
refinement of intuitionistic and classical logic that isolates the infinitary
parts of reasoning 
in two (dual) modalities: the \emph{exponentials} $\oc$
and $\wn$.
They give a logical status to the operations of memory management such as
\emph{copying} and \emph{erasing}: a linear proof corresponds---via Curry--Howard
isomorphism---to a program that uses its argument \emph{linearly}, \ie~exactly
once, while an exponential proof corresponds to a program that can use its
argument at will.



The intuition that linear programs are analogous to linear functions
(as studied in linear algebra) while exponential programs mirror a more general
class of analytic functions 
got a technical incarnation in Ehrhard's work \cite{Ehrhard:2002,Ehrhard:2005} on $\LL$-based denotational semantics for the $\lambda$-calculus.
This investigation has been then internalized in the syntax, yielding the \emph{resource $\lambda$-calculus} \cite{Boudol,Ehrhard:2003,Ehrhard:2008}: 
there, copying and erasing are forbidden 
and replaced by the possibility to apply a 
function to a \emph{bag} of resource $\lambda$-terms which specifies how many times an argument can be linearly passed to the function, so as to represent only bounded computations.

The \emph{Taylor expansion} associates with an ordinary $\lambda$-term
a (generally infinite) set of resource $\lambda$-terms, recursively approximating the usual application: the Taylor expansion of the $\lambda$-term $MN$ is made of resource $\lambda$-terms of the form $t[u_1, \dots, u_n]$, where
$t$ is a resource $\lambda$-term in the Taylor expansions of $M$, and $[u_1,\dots,u_n]$ is a bag of arbitrarily finitely many (possibly $0$) resource $\lambda$-terms in the Taylor expansion of	 $N$.
Roughly, the idea is to decompose a program into a set of
purely ``
resource-sensitive programs'', all of them containing only bounded (although
possibly non-linear) calls to inputs. The notion of Taylor expansion has many
applications in the theory of the $\lambda$-calculus, \eg in the study of linear
head reduction \cite{EhrhardRegnier06}, normalization \cite{DBLP:conf/fossacs/PaganiTV16,DBLP:conf/csl/Vaux17}, Böhm trees \cite{DBLP:conf/csl/BoudesHP13,KerinecMP18},
$\lambda$-theories~\cite{ManzonettoR14}, intersection
types~\cite{MazzaPellissierVial18}. More generally, understanding the relation
between a program and its Taylor expansion renews the logical approach to the
quantitative analysis of computation started with the inception of \LL.

A natural question is the \emph{inverse Taylor expansion problem}: how to
characterize which sets of resource $\lambda$-terms are contained in the Taylor
expansion of a same 
$\lambda$-term? Ehrhard and Regnier \cite{Ehrhard:2008} defined a simple \emph{coherence} relation such that a finite set of resource $\lambda$-terms is included in the Taylor expansion of a $\lambda$-term if and only if the elements
of this set are pairwise coherent. Coherence is crucial in many structural
properties of the resource $\lambda$-calculus, such as in the proof that in the $\lambda$-calculus normalization and Taylor expansion commute \cite{EhrhardRegnier06,Ehrhard:2008}.

We aim to solve the inverse Taylor expansion problem in the more general context of \LL, more precisely in the \emph{multiplicative-exponential fragment} $\MELL$ of \LL, being aware that for $\MELL$ 
no coherence relation can solve the problem (see below).


\subparagraph{Proof-nets, proof-structures and their Taylor expansion: seeing trees behind graphs}

In 
$\MELL$, linearity and the sharp analysis of computations naturally lead to represent proofs in a more general  \emph{graph}-like syntax instead of a term-like or tree-like one.\footnotemark
\footnotetext{A term-like object is essentially a tree, with one output (its root) and many inputs (its other leaves).} 
Indeed,  linear negation is involutive and classical duality can be interpreted as the possibility of juggling between different conclusions, without a distinguished output.
Graphs representing proofs in \MELL are called \emph{proof-nets}: their syntax is richer and more expressive than the $\lambda$-calculus. 
Contrary to $\lambda$-terms, proof-nets are special inhabitants of the wider land of \emph{proof-structures}: they can
be characterized, among proof-structures, by abstract (geometric) conditions
called correctness criteria \cite{Girard:1987}. The procedure of cut-elimination
can be applied to proof-structures, and proof-nets can also be seen as the
proof-structures with a good behavior with respect to cut-elimination
\cite{Bechet98}. Proof-structures can be interpreted in denotational models 
and proof-nets can be characterized among them by semantic means \cite{Retore97}.
It is then 
natural to attack problems in the general framework of proof-structures. 
In this work, correctness plays no role at all, hence we will
consider proof-structures and not only proof-nets. 
\MELL proof-structures are a
particular kind of graphs, whose edges are labeled by $\MELL$ formul\ae{} and
vertices by $\MELL$ connectives, and for which special subgraphs are
highlighted, the \emph{boxes}, representing the parts of the proof-structure
that can be copied and discarded (\ie called an unbounded number of times). 
A box is 
delimited from the rest of a proof-structure by exponential modalities: 
its border is made of one
$\oc$-cell, its principal door, and arbitrarily many $\wn$-cells, its auxiliary
doors.
Boxes are nested or disjoint (they cannot partially overlap), 
so as to add a tree-like structure to \mbox{proof-structures \emph{aside} from their graph-like nature.}

As in $\lambda$-calculus, 
one can define \cite{DBLP:journals/tcs/EhrhardR06} 
box-free \emph{resource proof-structures}\footnotemark,
\footnotetext{Also
  known as differential proof-structures \cite{Carvalho16} or differential nets \cite{DBLP:journals/tcs/EhrhardR06,DBLP:conf/lpar/MazzaP07,DBLP:journals/lmcs/Carvalho18} or simple nets \cite{PaganiTasson2009}.} where 
$\oc$-cells 
make resources available boundedly,
and the \emph{Taylor expansion} of $\MELL$ proof-structures into these resource proof-structures, that recursively copies the
content of the boxes an arbitrary number of times.
In fact, as somehow anticipated by Boudes \cite{Boudes:2009},
such a Taylor expansion operation can be carried on any tree-like
structure.
This primitive, abstract, notion of Taylor expansion can then be pulled back to
the structure of interest, \mbox{as 
	shown in \cite{Wollic} and put forth again here.}


\subparagraph*{The question of coherence for proof-structures}

The inverse Taylor expansion problem has a natural counterpart in the world of
\MELL proof-structures: given a set of resource proof-structures, is there a
$\MELL$ proof-structure the expansion of which contains the set? Pagani and
Tasson \cite{PaganiTasson2009} give the following answer: it is possible to
decide whether a finite set of resource proof-structures is a subset of the
Taylor expansion of a same \MELL proof-structure (and even possible to do it
in non-deterministic polynomial time); but unlike the $\lambda$-calculus, the structure of the relation ``being
part of the Taylor expansion of a same proof-structure'' is \emph{much more}
complicated than a binary (or even $n$-ary) coherence. Indeed, for any $n>1$, it
is possible to find ${n+1}$ resource proof-structures such that any $n$ of them
are in the Taylor expansion of some $\MELL$ proof-structure, but there is no
$\MELL$ proof-structure whose Taylor expansion has all the $n\!+\!1$ as elements
(see our \Cref{ex:not-coherent} and \cite[pp.~244-246]{Tasson:2009}).

In this work, we introduce a new combinatorial criterion, \emph{glueability},
for deciding whether a set of resource proof-structures is a subset of the
Taylor expansion of 
some $\MELL$ proof structure, based on a rewriting
system on sequences of $\MELL$ formul\ae{}. Our criterion is more general (and,
we believe, simpler) than the one of \cite{PaganiTasson2009}, which is limited
to the \emph{cut-free} case with \emph{atomic axioms} and characterizes only \emph{finite} sets: we do not have these limitations. 
We believe that our criterion is a useful tool for studying proof-structures.
We conjecture that it can be used to show that, for a suitable geometric
restriction, a binary coherence relation does exist for resource proof-structures. It
might also shed light on correctness~and~sequentialization.

As the proof-structures we consider are typed, an unrelated difficulty arises: a
resource proof-structure might not be in the Taylor expansion of any \(\MELL\)
proof-structure, not because it does not respect the structure imposed by the
Taylor expansion, but because its type is impossible.\footnotemark
\footnotetext{
Similarly, in the $\lambda$-calculus, there is no closed $\lambda$-term of type \(X \to Y\) with \(X \neq Y\) atomic, but the resource $\lambda$-term \((\lambda f.f)[\,]\) can be given that type: the empty bag $[\,]$ kills any information on the argument.} 
To solve this issue we enrich the \(\MELL\) proof-structure 
syntax with a ``universal'' proof-structure: 
a special \(\maltese\)-cell (\emph{daimon}) that can have any number of outputs of any types, and we allow it to appear inside a box, representing information plainly missing (see \Cref{sec:conclusions} for more details and the way this matter is handled by Pagani and Tasson~\cite{PaganiTasson2009}).

\section{Outline and technical issues}

\subparagraph*{The rewritings}
The essence of our rewriting system is not located on proof-structures but on
lists of \MELL formul\ae{} (\Cref{def:unwinding-paths}). In a very down-to-earth way, this rewriting system is
generated by elementary steps akin to rules of sequent calculus read from the
\emph{bottom up}: they act on a list of conclusions, analogous to a monolaterous
right-handed sequent. These steps are actually more sequentialized than sequent
calculus rules, as they 
do not allow for commutation. For instance, the
rule corresponding to the introduction of a $\otimes$ on the $i$-th formula, is
defined as
$\otimes_i : (\gamma_1,\dots,\gamma_{i-1},A\otimes B, \gamma_{i+1}, \dots,
\gamma_n) \to (\gamma_1,\dots,\gamma_{i-1},A, B, \gamma_{i+1}, \dots,
\gamma_n)$.

\begin{wrapfigure}[4]{r}{3cm}
  \vspace{-3.5\baselineskip}
  \scalebox{0.8}{
    \centering
    $
    \vspace{\beforepn}
    \tikzsetnextfilename{images/mell-tensor-intro-1}
    \pnet{
      \pnformulae{
        \pnf[A]{$A$}~\pnf[Ab]{$A^{\bot}$}
      }
      \pnaxiom[ax]{A,Ab}
      \pntensor{A,Ab}[i]{$A \otimes A^{\bot}$}
    }
    \pnrewrite[\otimes_1]
    \tikzsetnextfilename{images/mell-tensor-intro-2}
    \pnet{
      \pnformulae{
        \pnf[A]{$A$}~\pnf[Ab]{$A^{\bot}$}
      }
      \pnaxiom[ax]{A,Ab}
    }$}
\end{wrapfigure}
These rewrite steps then act on $\MELL$ proof-structures, coherently with
their type, by modifying (most of the times, erasing) the cells directly connected to the conclusion of the
proof-structure. Formally, this means that there is a functor $\qMELL$ from the
rewrite steps into the category $\Rel$ of sets and relations, associating with
a list of formul\ae{} the set of \MELL proof-structures with these conclusions, and with
a rewrite step a relation implementing it (\Cref{def:FunctorqMELL}). The rules
\emph{deconstruct} the proof-structure, starting from its conclusions. The rule
$\otimes_1$ acts by removing a $\otimes$-cell on the first conclusion, replacing
it by two conclusions.

These rules can only act on specific proof-structures, and indeed, capture a lot
of their structure: $\otimes_i$ can be applied to a \MELL proof-structure $R$ if and
only if $R$ has a $\otimes$-cell in the conclusion $i$ (as opposed to, say, an
axiom). So, in particular, every proof-structure is completely characterized by
any sequence rewriting it to the empty proof-structure.

\subparagraph{Naturality}
The same rules act also on sets of resource proof-structures, defining the
functor $\PPolyPN$ from the rewrite steps into the category $\Rel$
(\Cref{def:FunctorPPoly}). When carefully defined, the Taylor expansion induces
a \emph{natural transformation} from $\PPolyPN$ to $\qMELL$
(\Cref{thm:projection-natural}). By applying this naturality repeatedly, we get
our characterization (\Cref{thm:characterization}): a set of resource
proof-structures $\Pi$ is a subset of the Taylor expansion of a $\MELL$
proof-structure iff there is a sequence rewriting $\Pi$ to the singleton of
the~\emph{empty}~proof-structure.

The naturality property is not only a mean to get our characterization, but also an interesting result in itself: natural transformations can often be used to express fundamental properties in a mathematical context. 
In this case, the \emph{Taylor expansion is natural} with respect to the
possibility to build a proof-structure (both \(\MELL\) or resource) by adding a
cell to its conclusions or boxing it.
Said differently, naturality of the Taylor expansion roughly means that the rewrite rules that deconstruct a \MELL proof-structure $R$ and a set of resource proof-structures in the Taylor expansion of $R$ mimic each other.


\subparagraph*{Quasi-proof-structures and mix}
Our rewrite rules consume proof-structures from their conclusions. The rule
corresponding to boxes in $\MELL$ opens a box by deleting its principal door (a
$\oc$-cell) and its border, while for a resource proof-structure  
it deletes a $\oc$-cell and separates the different copies of the content of the box (possibly) represented by such a $\oc$-cell. 
This operation is problematic in a twofold way.
In a resource proof-structure, where the border of boxes is not marked, it is not clear how to identify such copies.
	On the other side, in a $\MELL$ proof-structure the content of a box is not to be treated as if it were at the same level as what is outside of the box: 
it can be copied many times or erased, while what is outside boxes cannot,
and treating the content in the same way as the outside suppresses this
distinction, which is crucial in \LL.
So, we need to remember that the content of a box, even if it is at depth $0$ (\ie not contained in any other box) after erasing the box wrapping it by means of our rewrite rules, 
is not to be mixed with the rest of the structure 
at depth $0$. 

\begin{wrapfigure}[3]{r}{1.75cm}
	\vspace{-\baselineskip}
	\scalebox{0.8}{
		\centering
		$
		\vspace{\beforepn}
		\tikzsetnextfilename{images/root}
		\pnet{
			\pnsomenet[u]{$\pi$}{1.25cm}{0.5cm}
			\pnoutfrom{u.-140}[pi1c]{$\qquad \ \!\!\cdots$}
			\pnoutfrom{u.-40}[u2a2]{$\quad$}
			\pnsemibox{u}
		}$}
\end{wrapfigure}
In order for our 
proof-structures to 
provide this information,
we need to generalize them and consider that a proof-structure can have not
just a tree of boxes, but a \emph{forest}: this yields the notion of
\emph{quasi-proof-structure} (\Cref{def:proof-structure}).
In this way, according to our rewrite rules, opening a box by deleting its principal door amounts to taking a
box in the tree and disconnecting it from its root, creating a new tree. We draw
this by surrounding elements having the same root with a dashed line, 
open from
the bottom, remembering the phantom presence of the border of the box, below,
even if it was erased. This allows one to open the box only when it is ``alone''
(see \Cref{def:unwinding}).

This is not merely a technical remark, as this generalization gives a status to
the $\Mix$ rule of \LL: indeed, mixing two proofs amounts to taking two proofs
and considering them as one, without any other modifications. Here, it amounts to
taking two proofs, each with its box-tree, and considering them as one by
merging the roots of their trees (see the mix step in \Cref{def:unwinding}). We
embed this design decision up to the level of formul\ae{}, which are segregated
in different zones that have to be mixed before interacting (see the notion of
partition of a finite sequence of formul\ae in \Cref{subparagraph:Formulas}).

\subparagraph{Geometric invariance and emptiness: the filled Taylor expansion}
\label{subsec:intro-fattened}
The use of forests instead of trees for the nesting structure of boxes, where
the different roots are thought of as the contents of long-gone boxes, has an
interesting consequence in the Taylor expansion: indeed, an element of the
Taylor expansion of a proof-structure contains an arbitrary number of copies of
the contents of the boxes, in particular \emph{zero}. If we think of the part at
depth $0$ of a $\MELL$ proof-structure as inside an invisible box, its content
can be deleted in some elements of the Taylor expansion just as any other
box.\footnotemark
\footnotetext{The dual case, of copying the contents of a box, poses no problem
  in our approach. 
}
As erasing completely conclusions would
cause the Taylor expansion not preserve the conclusions (which would lead to
technical complications), we introduce the \emph{filled Taylor expansion}
(\Cref{def:FilledTaylor}), which contains not only the elements of the usual
Taylor expansion, but also elements of the Taylor expansion where one component
has been erased and replaced by a \(\maltese\)-cell (\emph{daimon}), representing a lack of information, apart from the number and types of the conclusions.

\subparagraph{Atomic axioms}
Our paper first focuses on the case where proof-structures are restricted to \emph{atomic axioms}.
In \Cref{sec:non-atomic} we sketch how to adapt our method to the
non-atomic case.


%% file: proof-nets.tex
\section{Proof-structures and the Taylor expansion}
\label{sec:proof-nets}

\subparagraph{\texorpdfstring{$\MELL$}{MELL} formul\ae{} and (quasi-)proof-structures }
\label{subparagraph:Formulas}

Given a countably infinite set of propositional variables $X, Y, Z, \dots$, 
\emph{\MELL formul\ae{}} are defined by the following inductive grammar:
\begin{equation*}
  \label{eq:formula}
  A,B \Coloneqq X \mid X^\bot \mid \One \mid \bot \mid A
  \otimes B \mid A \parr B \mid \oc A \mid \wn A
\end{equation*}

Linear negation is defined via De Morgan laws $\One^\bot = \bot$,
$(A \otimes B)^\bot = A^\bot \parr B^\bot$ and $(\oc A)^\bot = \wn A$, so as to
be involutive, \ie $A^{\bot\bot} = A$.
Given a list $\Gamma = (A_1, \dots, A_m)$ of \MELL formul\ae, a \emph{partition}
of $\Gamma$ is a list $(\Gamma_1, \dots, \Gamma_n)$ of lists of \MELL formul\ae
such that there are $0 = i_0 < \dots < i_n = m$ with
\mbox{$\Gamma_j = (A_{i_{j-1}+1}, \dots, A_{i_{j}})$} for all
$1 \leqslant j \leqslant n$; such a partition of $\Gamma$ is also denoted by
$(A_1, \dots, A_{i_1}; \cdots ; A_{i_{n-1}+1}, \dots, A_m)$, 
with lists separated by semi-colons. 


\begin{figure}[!t]
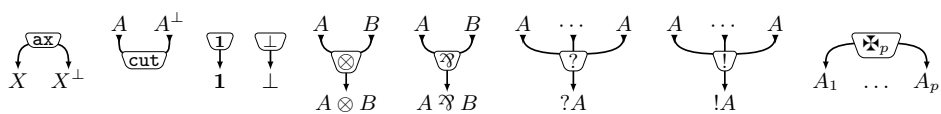

  \vspace*{\beforepn}
  \scalebox{0.9}{
    \centering 
    \tikzsetnextfilename{images/resource-cells}
    \pnet{
      \pnformulae{
        ~~
        \pnf[Xcut]{$A$}~\pnf[Xncut]{$A^{\bot}$}~~~
        \pnf[Ao]{$A$}~\pnf[Bo]{$B$}~
        \pnf[Ap]{$A$}~\pnf[Bp]{$B$}~
        \pnf[A1]{$A$}~\pnf[l]{$\cdots$}~\pnf[An]{$A$}~
        \pnf[A1e]{$A$}~\pnf[le]{$\cdots$}~\pnf[Ane]{$A$}
        \\
        \pnf[Xax]{$X$}~\pnf[Xnax]{$X^{\bot}$}~~~
        \pnf[1]{$\mathbf{1}$}~\pnf[bot]{$\bot$}~~~~~~~~~~~
        \pnf[i]{$A_1$}~\pnf[d]{$\ldots$}~\pnf[j]{$A_p$}
      }
      \pnaxiom[ax]{Xax,Xnax} 
      \pncut[cut]{Xcut,Xncut}
      \pnone{1}
      \pnbot{bot}
      \pntensor{Ao,Bo}[AoB]{$A\otimes B$}
      \pnpar{Ap,Bp}[ApB]{$A \parr B$}
      \pnexp{}{A1,l,An}[]{$\wn A$}
      \pnbag{}{A1e,le,Ane}[]{$\oc A$}
      \pndaimonp{i,j}
    }
  }
	\vspace*{-\baselineskip}
  \caption{Cells, with their labels and their typed inputs and outputs (ordered from  left to right).}
  \label{fig:resource-cells}
\end{figure}


We reuse the syntax of proof-structures given in \cite{Wollic} and sketch here
its main features. We suppose known definitions of (directed) graph, rooted tree, and
morphism of these structures. In what follows we will speak of \emph{tails} in
a graph: ``hanging'' edges with only one vertex. This can be implemented either
by adding special 
vertices or using \cite{Borisov2008}'s graphs.

If an edge $e$ is incoming in (resp.~outgoing from) a vertex $v$, we say that
$e$ is a \emph{input} (resp.~\emph{output}) of $v$. The reflexive-transitive
closure of a tree $\tau$ is denoted by $\tau^{\ReflexiveTransitive}$: the
operator $(\cdot)^{\ReflexiveTransitive}$ lifts to a functor from the category
of trees to the category of directed graphs.

\begin{definition}
\label{def:proof-structure}
  A \emph{module} $M$ is a (finite) directed graph with:
  \begin{itemize}
  \item vertices $v$ labeled by
    $\VertType(v) \in \{ \ax, \cut, \One, \bot, \otimes, \parr, \wn, \oc\} \cup
    \{\maltese_p \mid p \in \Nat\}$, the \emph{type} of $v$; 
  \item edges $e$ labeled by a $\MELL$ formula $\FlagType(e)$, the \emph{type} of $e$; 
  \item an order $<_M$ 
that is total on the tails of $|M|$ 
and on the inputs of
each vertex of~type~$\parr, \otimes$. 
  \end{itemize}
  Moreover, all the vertices verify the conditions of
  \Cref{fig:resource-cells}.\footnotemark
  \footnotetext{Note that there are no conditions on the types of the outputs of 
    vertices of type $\maltese$ (\ie~of type $\maltese_p$ for some $p \in \Nat$); and the outputs of vertices of type $\ax$ must have \emph{atomic} types.}

  A \emph{quasi-proof-structure} is a triple $R = (|R|,\TreeT,\BoxFunction)$ where:
  \begin{itemize}
  \item $|R|$ is a module with no input tails, called the \emph{module} of $R$;
  \item $\TreeT$ is a forest of rooted trees with no input tails, called the
    \emph{box-forest} of $R$;
  \item $\BoxFunction \colon |R| \to \TreeT^{\ReflexiveTransitive}$ is a morphism of directed graphs, the \emph{box-function} of $R$, which induces a partial bijection from the inputs of the vertices
    of type $\oc$
    and the 
    edges in $\TreeT$,
    and such~that:
  \begin{itemize}
		\item for any vertices $v, v'$ with an edge from $v'$ to $v$, 
		if $\BoxFunction(v) \neq \BoxFunction(v')$ then $\VertType(v) \in \{ \oc, \wn \}$.\footnotemark
		\footnotetext{Roughly, it says that the border of a box is made of (inputs of) vertices of type $\oc$ or $\wn$.}
\end{itemize}
  	Moreover, for any output tails $e_1, e_2, e_3$ in $|R|$ which are outputs of the vertices $v_1, v_2, v_3$, respectively, if $e_1 <_{|R|} e_2 <_{|R|} e_3$ then it is impossible that $\BoxFunction(v_1) = \BoxFunction(v_3) \neq \BoxFunction(v_2)$.\footnotemark
  	\footnotetext{This is a technical condition that simplifies the definition of the rewrite rules in \Cref{sec:rules}. 
  		Note that $\BoxFunction(v_1), \BoxFunction(v_2), \BoxFunction(v_3)$ are necessarily roots in $\TreeT$, since $\BoxFunction$ is a morphism of directed graphs.}
%
  \end{itemize}

  A quasi-proof-structure $R = (|R|, \TreeT, \BoxFunction)$ is:
 	\begin{enumerate}
 		\item $\MELLdaimon$ if
all vertices in $|R|$ of type $\oc$ have exactly one input, and
	the partial bijection induced by $\BoxFunction$ from the inputs of the vertices
	    of type $\oc$ in $|R|$ 
	    and the 
	    edges in $\TreeT$ is 
	    total.

		\item $\MELL$ if it is $\MELLdaimon$ and, for every vertex $v$ in $|R|$ of type $\maltese$, one has $\BoxFunction^{{-}1}(\BoxFunction(v)) = \{v\}$ and $\BoxFunction(v)$ is not a root of the box-forest $\TreeT$ of $R$.
	
		\item $\DiLL_0^\maltese$ if the box-forest $\TreeT$ of $R$ is just a juxtaposition of roots.
		
		\item $\DiLL_0$ (or \emph{resource}) if it is $\DiLL_0^\maltese$ and there is no vertex in $|R|$ of type $\maltese$.
	\end{enumerate}
	For the previous systems, a \emph{proof-structure} is a quasi-proof-structure whose box-forest~is~a~tree.
\end{definition}

Our $\MELL$ proof-structure (\ie~a $\MELL$ quasi-proof-structure that is also a proof-structure) corresponds to the usual notion of $\MELL$ proof-structure (as in \cite{CarvalhoTortora}) except that we also allow the presence of a box filled only by a \emph{daimon} (\ie a vertex of type $\maltese$).
The \emph{empty} ($\DiLL_0$ and \MELL) proof-structure---whose module and box-forest are empty graphs---is denoted~by~$\emptynet$.

Given a quasi-proof-structure $R = (|R|, \TreeT, \BoxFunction)$, the output tails of $|R|$ are the \emph{conclusions} of $R$. 
So, the pre-images of the roots of $\TreeT$ via $\BoxFunction$ 
partition the conclusions of $R$ in a list of lists of such conclusions.
The \emph{type} of $R$ is the list of lists of the types of these
conclusions.
We often identify the conclusions of $R$ with a finite initial segment of $\Nat$.

By definition of graph morphism, two conclusions in two
distinct lists in the type of a quasi-proof-structure $R$ are in two distinct
connected components of $|R|$; so, if $R$ is not a proof-structure then $|R|$
contains several connected components. 
Thus, $R$ 
can be seen as a list of proof-structures, its
\emph{components}, one for each root in its box-forest.

A non-root vertex $v$ in the box-forest $\TreeT$ induces a subgraph of
$\TreeT^{\ReflexiveTransitive}$ of all vertices above it and edges connecting
them. The pre-image of this subgraph through $\BoxFunction$ is the
\emph{box} of $v$ and the conditions on $\BoxFunction$ in \Cref{def:proof-structure} translate 
the usual nesting condition for \LL boxes.


In 
quasi-proof-structures, we speak of \emph{cells} instead of
vertices, and, for a cell of type $\ell$, of a \emph{$\ell$-cell}.
A $\maltese$-cell is a $\maltese_{p}$-cell for some $p \in \Nat$. 
An \emph{hypothesis cell} is a cell without inputs.



\begin{example}
  The graph in \Cref{fig:pointed-proof-net} is a \MELL quasi-proof-structure.
  The colored areas represent the pre-images of boxes, and the dashed boxes
  represent the pre-images of roots.
\end{example}

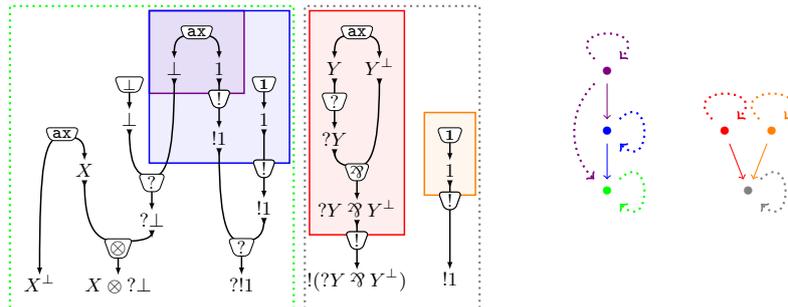
\begin{figure}[!th]
  \vspace{\beforepn}
  \centering
  \tikzsetnextfilename{images/church-3}
  \scalebox{0.8}{ 
    \pnet{
      \pnformulae{
        ~~~\pnf[botax]{$\bot$}~\pnf[1ax]{$1$}~~~\pnf[Y]{$Y$}~\pnf[Yp]{$Y^\bot$} \\
        ~~\pnf[bot]{$\bot$}~~~\pnf[1]{$1$} \\
        ~\pnf[X]{$X$}~~~~~~~~~\pnf[1']{$1$} \\ \\ \\
        \pnf[Xp]{$X^\bot$}
      }
      \pnaxiom[ax]{botax,1ax}
      \pnaxiom[axprop]{Xp,X}
      \pnaxiom[axprop']{Y,Yp}
      \pnboxc{ax,botax,1ax}{violet}
      \pnprom{1ax}[!1ax]{$\oc 1$}
      \pnbot[botcell]{bot}
      \pnone[1cell]{1}
      \pnboxc{1,1cell,ax,botax,1ax,!1ax}{blue}
      \pnprom{1}[!1]{$\oc 1$}
      \pnexp{}{bot,botax}[?bot]{$\wn \bot$}[1.5]
      \pnexp{}{!1ax,!1}[?!1]{$\wn\oc 1$}[1.1]
      \pntensor{X,?bot}[tens]{$X \otimes \wn \bot$}[.9]
      \pnexp{}{Y}[?Y]{$\wn Y$}
      \pnpar{?Y,Yp}[par]{$\wn Y \parr Y^\bot$}
      \pnboxc{axprop',Y,Yp,?Y,par}{red}
      \pnprom{par}[!par]{$\oc (\wn Y \parr Y^\bot)$}
      \pnone[1cell']{1'}
      \pnboxc{1',1cell'}{orange}
      \pnprom{1'}[!1']{$\oc 1$}[2]
      \pnsemiboxc{axprop,Xp,1,ax,?!1}{green}
      \pnsemibox{axprop',!1',Y}
    }
    \qquad\quad
    \raisebox{-1cm}{
      \begin{forest}
        for tree={grow'=north, edge=<-}
        [$\textcolor{green}{\bullet}$, name = root1,  edge = green
        [$\textcolor{blue}{\bullet}$, name = box11,  edge = blue
        [$\textcolor{violet}{\bullet}$, name = box111, edge = violet]
        ]
        ]
        \draw[->,dotted,violet,line width=1pt] (box111) to[out=south west,
        in=north west] (root1);
        \draw[->,dotted,violet,line width=1pt,loop,looseness=4.8] (box111)
        to[out=north west, in=north east] (box111);
        \draw[->,dotted,blue,line width=1pt,loop,looseness=4.8] (box11)
        to[out=north east,in=south east] (box11);      
        \draw[->,dotted,green,line width=1pt,loop,looseness=4.8] (root1)
        to[out=north east,in=south east] (root1);
      \end{forest}
      \ \ 
      \begin{forest}
        for tree={grow'=north, edge=<-, s sep=3.5mm}
        [$\textcolor{gray}{\bullet}$, name = root2
        [$\textcolor{red}{\bullet}$, name = box21,  edge = red]
        [$\textcolor{orange}{\bullet}$, name = box22,  edge = orange]
        ]
        \draw[->,dotted,red,line width=1pt,loop,looseness=4.8] (box21)
        to[out=north west,in=north east] (box21);
        \draw[->,dotted,orange,line width=1pt,loop,looseness=4.8] (box22)
        to[out=north west,in=north east] (box22);
        \draw[->,dotted,gray,line width=1pt,loop,looseness=4.8] (root2)
        to[out=north east,in=south east] (root2);
      \end{forest}}
  }
	\vspace{-\baselineskip}
	\caption{A \MELL quasi-proof-structure $R$, its box-forest
          \texorpdfstring{$\TreeT_{R}$}{} (without dotted lines) and the
          reflexive-transitive closure
          \texorpdfstring{$\TreeT_{R}^\circlearrowleft$}{} of
          \texorpdfstring{$\TreeT_{R}$}{} (with also dotted lines).}
	\label{fig:pointed-proof-net}
\end{figure}



%% file: taylor.tex
\subparagraph{The Taylor expansion}
Proof-structures have a tree structure made explicit by their
box-function. 
Following \cite{Wollic}, the definition of the Taylor expansion 
 uses this tree structure: first, we define how to ``\emph{expand}'' a tree---and more generally a forest---via a generalization of the notion of thick subtree \cite{Boudes:2009} (\Cref{def:thick}; roughly, a thick subforest of a box-forest says the number of copies of each box to be taken, iteratively), we then take all the expansions of the tree structure of a proof-structure and we \emph{pull} the approximations \emph{back} to the underlying graphs (\Cref{def:proto}), finally we \emph{forget} the tree structures associated with them (\Cref{def:taylor}). 

\begin{definition}[thick subforest]
\label{def:thick}
  \index{tree!thick subtree}	
  Let $\tau$ be a forest of rooted trees.
  A \emph{thick subforest} of $\tau$ is a pair $(\sigma, h)$ of a forest
  $\sigma$ of rooted trees and a graph morphism $h \colon \sigma \to \tau$ whose
  restriction to the roots of $\sigma$ is bijective.
\end{definition}

\begin{example}
  \label{ex:proto-taylor}
  The following is a graphical presentation of a thick subforest $(\tau,h)$ of
  the box-forest $\TreeT$ of the quasi-proof-structure in
  \Cref{fig:pointed-proof-net}, where the graph morphism
  $h \colon \tau \to \TreeT$ is depicted chromatically (same color means same
  image via $h$).
  \vspace{-0.3\baselineskip}
  \begin{align*}
		\scalebox{0.7}{
			\raisebox{1.3cm}{\Large$\tau \ = \ $}
			\begin{forest}
			for tree={grow'=north, edge=<-, l=6mm}
			[$\textcolor{green}{\bullet}$
				[$\textcolor{blue}{\bullet}$,  edge = blue
					[$\textcolor{violet}{\bullet}$, edge = violet]
				]
				[$\textcolor{blue}{\bullet}$,  edge = blue]
				[$\textcolor{blue}{\bullet}$,  edge = blue 
					[$\textcolor{violet}{\bullet}$,  edge = violet]
					[$\textcolor{violet}{\bullet}$,  edge = violet] 
				]
			]
			\end{forest}
			\
			\begin{forest}
			for tree={grow'=north, edge=<-, l=6mm}
			[$\textcolor{gray}{\bullet}$
				[$\textcolor{red}{\bullet}$,  edge = red]
				[$\textcolor{orange}{\bullet}$,  edge = orange]
				[$\textcolor{orange}{\bullet}$,  edge = orange]
				[$\textcolor{orange}{\bullet}$,  edge = orange] 
				[$\textcolor{orange}{\bullet}$,  edge = orange]
			]
			\end{forest}
			\raisebox{1.3cm}{\Large$\qquad\stackrel{h}{\longrightarrow}\qquad$}
			\begin{forest}
			for tree={grow'=north, edge=<-, l=6mm}
			[$\textcolor{green}{\bullet}$ 
				[$\textcolor{blue}{\bullet}$,  edge = blue
					[$\textcolor{violet}{\bullet}$,  edge = violet]
				]
			]
			\end{forest}
			\quad 
			\begin{forest}
			for tree={grow'=north, edge=<-, l=6mm}
			[$\textcolor{gray}{\bullet}$ 
				[$\textcolor{red}{\bullet}$,  edge = red]
				[$\textcolor{orange}{\bullet}$,  edge = orange]
			]
			\end{forest}
			\raisebox{1.3cm}{\Large$ \ =  \ \TreeT$}
		}
	\end{align*}
	Intuitively, it means that $\tau$ is obtained from $\TreeT$ by taking $3$ copies of the blue box, $1$ copy of the red box and $4$ copies of the orange box; in the first (resp.~second; third) copy of the blue box, $1$ copy (resp.~$0$ copies; $2$ copies) of the purple box has been taken.
\end{example}


\begin{definition}[proto-Taylor expansion]
\label{def:proto}
  Let $R = (|R|, \TreeT_R, \BoxFunction_{R})$ be a quasi-proof-structure. 
  The \emph{proto-Taylor expansion} of $R$ is the set $\ProtoTaylor{R}$ of thick
  subforests of $\TreeT_{R}$.
  
  Let $t = (\tau_t, h_t) \in \ProtoTaylor{R}$.
  The \emph{$t$-expansion} of $R$ is the pullback $(R_t, p_t, p_R)$ below, computed in the category of 
  directed graphs and graph morphisms.
  \tikzsetnextfilename{images/Taylor-pullback}
  \begin{commutativediagram}{3}{3}{0}
    {
      R_t \& \tau_t^{\circlearrowleft}\\
      \lvert R \rvert  \& \TreeT_{R}^{\circlearrowleft}\\
    };
    \path[-stealth]
    (m-1-1) edge [auto] node {$p_t$} (m-1-2)
    (m-1-1) edge [auto] node {$p_R$} (m-2-1)
    (m-1-2) edge [auto] node {$h_t^{\circlearrowleft}$} (m-2-2)
    (m-2-1) edge [auto] node {$\BoxFunction_{R}$} (m-2-2);
    \begin{scope}[shift=($(m-1-1)!.25!(m-2-2)$)]
      \draw +(-.35,0) -- +(0,0)  -- +(0,.28);
    \end{scope}
  \end{commutativediagram}
\end{definition}

Given a quasi-proof-structure $R$ and $t = (\tau_t, h_t) \in \ProtoTaylor{R}$,
the 
directed graph $R_t$ inherits labels on vertices and edges by
composition with the graph morphism $p_R \colon R_t \to |R|$.

Let $[\tau_t]$ be the forest made up of the roots of $\tau_t$ and
$\iota \colon \tau_t \to [\tau_t]$ be the graph morphism sending each vertex of
$\tau_t$ to the root below it; $\iota^\circlearrowleft$ induces by
post-composition a morphism
$\overline{h_t} = \iota^\circlearrowleft \circ p_t \colon R_t \to
[\tau_t]^{\circlearrowleft}$. The triple $(R_t, [\tau_t], \overline{h_t})$ is a
\polyadic quasi-proof-structure, and it is a \polyadic proof-structure if $R$ is a
proof-structure.
We can then define the \emph{Taylor expansion} $\Taylor{R}$ of a quasi-proof-structure $R$ (an example of an element of a Taylor expansion~is~in~\Cref{fig:taylor-expansion}).

\begin{definition}[Taylor expansion]
\label{def:taylor}
  Let $R$ be a quasi-proof-structure. The \emph{Taylor expansion} of $R$ is the
  set of \polyadic quasi-proof-structures
  $\Taylor{R} = \{ (R_t, [\tau_t], \overline{h_t}) \mid t = (\tau_t, h_t) \in
  \ProtoTaylor{R}\}$.
\end{definition}

An element $(R_t, [\tau_t], \overline{h_t})$ of the Taylor expansion of a quasi-proof-structure $R$ 
has much less
structure than the pullback $(R_t,p_t,p_R)$: the latter 
indeed is a \polyadic
quasi-proof-structure $R_t$ coming with its projections
$|R| \stackrel{p_R}{\longleftarrow} R_t \stackrel{p_t}{\longrightarrow}
\tau_t^\circlearrowleft$, which establish a precise correspondence between cells and edges of $R_t$ and cells and edges of $R$: 
a cell in $R_t$ is labeled (via 
the projections) by both the cell of $|R|$ and the branch of the box-forest of $R$ it
arose from.  But $(R_t, [\tau_t], \overline{h_t})$ where $R_t$ is without its
projections $p_t$ and $p_R$ loses the correspondence with
$R $. 

\begin{figure}[!t]
	\vspace{\beforepn}
	\centering
	\tikzsetnextfilename{images/taylor-3}
	\scalebox{0.8}{ 
		\pnet{
			\pnformulae{
				~~~\pnf[botax]{$\bot$}~\pnf[1ax]{$1$}~\pnf[botaxbis]{$\bot$}~\pnf[1axbis]{$1$}~\pnf[botaxtris]{$\bot$}~\pnf[1axtris]{$1$}~~~~~\pnf[Y]{$Y$}~\pnf[Yp]{$Y^\bot$} \\
				~~\pnf[bot]{$\bot$}~~~~~~~\pnf[1]{$1$}~\pnf[1bis]{$1$}~\pnf[1tris]{$1$} \\
				~\pnf[X]{$X$}~~~~~~~~~~~~~~\pnf[1']{$1$}~\pnf[1'']{$1$}~\pnf[1''']{$1$}~\pnf[1'''']{$1$} \\ 
				\pnf[Xp]{$X^\bot$}~~~~~~\pnf[1cow]{$\oc 1$}
			}
			\pnaxiom[ax]{botax,1ax}
			\pnaxiom[axprop]{Xp,X}
			\pnaxiom[axbis]{botaxbis,1axbis}
			\pnaxiom[axtris]{botaxtris,1axtris}
			\pnaxiom[axprop']{Y,Yp}
			\pnbag{}{1ax}[!1ax]{$\oc 1$}
			\pnbag{}{1axbis,1axtris}[!1axtris]{$\oc 1$}[1][.3]
			\pninitial{$\oc$}[cow]{1cow}
			\pnbot[botcell]{bot}
			\pnone[1cell]{1}
			\pnone[1cellbis]{1bis}			
			\pnone[1celltris]{1tris}
			\pnbag{}{1,1bis,1tris}[!1]{$\oc 1$}
			\pnexp{}{bot,botax,botaxbis,botaxtris}[?bot]{$\wn \bot$}[1.5][-1]
			\pnexp{}{!1ax,!1axtris,!1,1cow}[?!1]{$\wn\oc 1$}[1]
			\pntensor{X,?bot}[tens]{$X \otimes \wn \bot$}[.9]
			\pnexp{}{Y}[?Y]{$\wn Y$}
			\pnpar{?Y,Yp}[par]{$\wn Y \parr Y^\bot$}
			\pnbag{}{par}[!par]{$\oc (\wn Y \parr Y^\bot)$}
			\pnone[1cell']{1'}
			\pnone[1cell'']{1''}
			\pnone[1cell''']{1'''}
			\pnone[1cell'''']{1''''}
			\pnbag{}{1',1'',1''',1''''}[!1']{$\oc 1$}[1.5]
			\pnsemiboxc{axprop,Xp,1,1bis,1tris,ax,?!1}{green}
			\pnsemibox{axprop',!1',1'''',Y}
		}
		\qquad
		\begin{forest}
			for tree={grow'=north}
			[$\textcolor{green}{\bullet}$, name=root1]
		\end{forest}
		\ \ 
		\begin{forest}
			for tree={grow'=north}
			[$\textcolor{gray}{\bullet}$, name=root2]
		\end{forest}
	}
	\caption{The element of the Taylor expansion of the \MELL
          quasi-proof-structure $R$ in \cref{fig:pointed-proof-net}, obtained
          from the element of $\ProtoTaylor{R}$ depicted in
          \cref{ex:proto-taylor}.}
	\label{fig:taylor-expansion}
\end{figure}

\begin{remark}\label{rmk:conclusions}
	By definition, the Taylor expansion preserves conclusions: there is a bijection $\varphi$ from the conclusions of a quasi-proof-structure $R$ to the ones in each element $\rho$ of $\Taylor{R}$ such that $i$ and $\varphi(i)$ have the same type and the same root (\ie $\BoxFunction_{R}(i) = \BoxFunction_{\rho}(\varphi(i))$ up to isomorphism).
	Therefore, the types of $R$ and $\rho$ are the same (as a list of lists).
\end{remark}


\subparagraph{The \filled Taylor expansion}
\label{subsec:fat}

As discussed in \Cref{subsec:intro-fattened}
(p.~\pageref{subsec:intro-fattened}), our method needs to 
``represent'' the emptiness introduced by the Taylor expansion (taking $0$
copies of a box) so as to preserve the 
conclusions. 
So, an element of the \emph{\filled Taylor expansion} $\FatTaylor{R}$ of a quasi-proof-structure $R$ 
(an example is in \Cref{fig:emptyings}) is obtained from an element of $\Taylor{R}$
where a whole component can be erased and replaced by a
\(\maltese\)-cell with the same conclusions \mbox{(hence $\Taylor{R} \subseteq \FatTaylor{R}$).}

\begin{definition}[filled Taylor expansion]\label{def:FilledTaylor}
%
An \emph{emptying} of a \polyadic quasi-proof-structure \(\rho = (|\rho|, \TreeT, \BoxFunction)\) is the \polyadic quasi-proof-structure with the same conclusions as in \(\rho\), obtained from $\rho$ 
by replacing each of the components of some roots 
of $\TreeT$ with a $\maltese$-cell whose outputs~are~tails.

  The \emph{\filled Taylor expansion} \(\FatTaylor{R}\) of a quasi-proof-structure $R$ is the set of all the emptyings of every element of its Taylor expansion $\Taylor{R}$.
\end{definition}

\begin{figure}[!h]
  \vspace{\beforepn}
  \centering
  \scalebox{0.75}{
    \tikzsetnextfilename{images/emptying-1}
    \pnet{
      \pnformulae{
        ~~~~~~\pnf[yg]{}~~\pnf[1']{$1$}~\pnf[1'']{$1$}\\
        \pnf[Xp]{$X^\bot$}~~\pnf[tens]{$X\otimes \wn \bot$}~~\pnf[?!1]{$\wn\oc 1$}~~~\pnf[YY]{$\oc (\wn Y \parr Y^\bot)$}
      }
      \pndaimon[dai]{Xp,tens,?!1}
      \pncown[!par]{YY}
      \pnone[1cell']{1'}
      \pnone[1cell'']{1''}
      \pnbag{}{1',1''}[!1']{$\oc 1$}
      \pnsemiboxc{dai,Xp,tens,?!1}{green}
      \pnsemibox{YY,1cell',1cell'',!par,yg}
    }
  } \qquad
	\
	\begin{forest}
		for tree={grow'=north}
		[$\textcolor{green}{\bullet}$]
	\end{forest}
	\begin{forest}
		for tree={grow'=north}
		[$\textcolor{gray}{\bullet}$]
	\end{forest}
	\quad
  \caption{An element of the \filled Taylor expansion of the \MELL quasi-proof-structure in \cref{fig:pointed-proof-net}.}
  \label{fig:emptyings}
\end{figure}
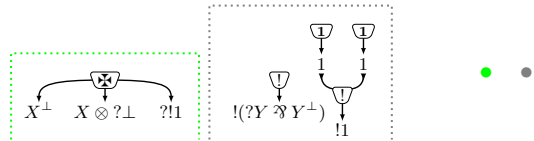


%% file: rules.tex
\section{Means of destruction: unwinding \texorpdfstring{$\MELL$}{MELL} quasi-proof-structures}
\label{sec:rules}

\begin{figure}[t]
	\centering
	\begin{align*}
	\scalebox{0.8}{
        $\begin{array}{rcl}
	(\Gamma_1; \cdots; \Gamma_k, \FlagType(i), \FlagType(i\!+\!1), \Gamma_{k}'; \cdots;	\Gamma_n)&\xrightarrow{\Ex_{i}}& (\Gamma_1; \cdots; \Gamma_k, \FlagType(i\!+\!1),  \FlagType(i), \Gamma_{k}'; \cdots;	\Gamma_n)\\
      	(\Gamma_1; \cdots; \Gamma_{k}, \FlagType(i), \FlagType(i\!+\!1), \Gamma_{k}'; \cdots ; \Gamma_n) &\xrightarrow{\Mix_{i}} & (\Gamma_1; \cdots; \Gamma_{k}, \FlagType(i); \FlagType(i\!+\!1), \Gamma_{k}'; \cdots ; \Gamma_n) \\
 	(\Gamma_1; \cdots; \Gamma_{k}; \FlagType(i), \FlagType(i\!+\!1); \Gamma_{k+2}; \cdots; \Gamma_n)
	&\xrightarrow{\Ax_{i}}& (\Gamma_1; \cdots; \Gamma_{k};	\Gamma_{k+2}; \cdots; \Gamma_n) \quad \text{with } \FlagType(i) = A = \FlagType(i\!+\!1)^\bot \\
	(\Gamma_1; \cdots; \Gamma_{k}; \cdots; \Gamma_n) &\xrightarrow{\Cut^i} &(\Gamma_1; \cdots; \Gamma_k, \FlagType(i), \FlagType(i\!+\!1); \cdots; \Gamma_n) \quad \text{with } \FlagType(i) = A = \FlagType(i\!+\!1)^\bot \\
	(\Gamma_1; \cdots; \Gamma_{k}; \Gamma_{k+1}, \FlagType(i); \Gamma_{k+2}; \cdots; \Gamma_n)
	&\xrightarrow{\maltese_{i}}
	&(\Gamma_1; \cdots; \Gamma_{k}; \Gamma_{k+2}; \cdots; \Gamma_n)\\
	(\Gamma_1; \cdots; \Gamma_k; \FlagType(i); \Gamma_{k+2}; \cdots; \Gamma_n) &\xrightarrow{\One_i}
	&(\Gamma_1; \cdots; \Gamma_k; \Gamma_{k+2}; \cdots; \Gamma_n) \quad \text{with } \FlagType(i) = \One \\
	(\Gamma_1; \cdots; \Gamma_k; \FlagType(i); \Gamma_{k+2}; \cdots; \Gamma_n) &\xrightarrow{\bot_i}
	&(\Gamma_1; \cdots; \Gamma_k; \Gamma_{k+2}; \cdots; \Gamma_n) \quad \text{with } \FlagType(i) = \bot \\
	(\Gamma_1; \cdots; \Gamma_k, \FlagType(i); \cdots; \Gamma_n) &\xrightarrow{\otimes_i}
	&(\Gamma_1; \cdots; \Gamma_{k}, A, B; \cdots;
	\Gamma_n) \quad \text{with } \FlagType(i) = A \otimes B \\
	(\Gamma_1; \cdots; \Gamma_{k}, \FlagType(i); \cdots; \Gamma_n) &\xrightarrow{\parr_i}
	&(\Gamma_1; \cdots; \Gamma_{k}, A, B; \cdots;
	\Gamma_n) \quad \text{with } \FlagType(i) = A \parr B \\
	(\Gamma_1; \cdots; \Gamma_{k}, \FlagType(i); \cdots; \Gamma_n) &\xrightarrow{\Contr_i}
	&(\Gamma_1; \cdots; \Gamma_{k}, \wn A, \wn A;
	\cdots, \Gamma_n) \quad \text{with } \FlagType(i) = \wn A 
	\\
	(\Gamma_1; \cdots; \Gamma_{k}, \FlagType(i); \cdots; \Gamma_n) &\xrightarrow{\Der_i}
	&(\Gamma_1; \cdots; \Gamma_{k},
	A; \cdots; \Gamma_n) \quad \text{with } \FlagType(i) = \wn A 
	\\
	(\Gamma_1; \cdots; \Gamma_{k}; \FlagType(i); \Gamma_{k+2}; \cdots; \Gamma_n) &\xrightarrow{\Weak_i}
	&(\Gamma_1; \cdots; \Gamma_{k}; \Gamma_{k+2}; \cdots; \Gamma_n) \quad \text{with } \FlagType(i) = \wn A 
	\\
	(\Gamma_1; \cdots; \wn \Gamma_k, \FlagType(i);  \cdots; \Gamma_n) &\xrightarrow{\BoxR_{i}}
	&(\Gamma_1; \cdots; \wn \Gamma_k, A; \cdots;	\Gamma_n) \quad \text{with } \FlagType(i) = \oc A
	\end{array}$
	}
	\end{align*}
	\vspace*{-\baselineskip}
	\caption{The generators of $\Path$. 
	In the source $\Gamma = (A_1, \dots, A_{i_1}; \cdots ; A_{i_{m-1}+1}, \dots, A_{i_n})$ 
	of each arrow, $\FlagType(i)$ denotes the $i^\text{th}$ formula in the flattening $(A_1, \dots, A_{i_1}, \dots, A_{i_{m-1}+1}, \dots, A_{i_n})$ of $\Gamma$.
	}
	\label{fig:elementary-path}
\end{figure}

Our aim is to deconstruct proof-structures (be they $\MELLdaimon$ or \polyadic) from
their conclusions. To do that, we introduce a category of rules of
deconstruction. The morphisms of this category are sequences of deconstructing
rules, acting on lists of lists of formul\ae{}. These morphisms act through
functors on quasi-proof-structures, exhibiting their sequential structure.

\begin{definition}[the category $\Path$]
  \label{def:unwinding-paths}
  Let $\Path$ be the category whose
  \begin{itemize}
  \item objects are lists $\Gamma = (\Gamma_1; \dots; \Gamma_n)$ of lists of
    $\MELL$ formul\ae{};
  \item arrows are freely generated by the \emph{elementary paths} in
    \Cref{fig:elementary-path}.
  \end{itemize}
  We call a \emph{path} any arrow $\xi \colon \Gamma \to \Gamma'$.  We write the
  composition of paths without symbols and in the diagrammatic order, so, if
  $\xi \colon \Gamma \to \Gamma'$ and $\xi' \colon \Gamma' \to \Gamma''$,
  $\xi\xi' \colon \Gamma\to \Gamma''$.
\end{definition}

\begin{example}
  \label{ex:barbara-path}
  $\parr_1\parr_2\parr_3 \otimes_1\otimes_3 \Ex_1\Ex_2\Mix_{2}\Ax_1
  \Ex_2\Mix_{2}\Ax_1 \Ax_1$ is a path of type
  $\big((X \otimes Y^\bot) \parr ((Y \otimes Z^\bot) \parr (X^\bot \parr Z))\big)
  \longrightarrow \emptynet$, where $\emptynet$ is the \emph{empty list} of lists of formul\ae.
\end{example}

We will tend to forget about exchanges and perform them silently (as it is
customary, for instance, in most presentations of sequent calculi).


%% file: unwinding.tex
\label{sec:unwinding}

The category $\Path$ acts on $\MELLdaimon$ quasi-proof-structures, exhibiting a sequential structure in their construction. 
For $\Gamma$ a list of list of $\MELL$ formul\ae{}, $\qMELL(\Gamma)$ is the set of $\MELLdaimon$ quasi-proof-structures of type $\Gamma$.
To ease the reading of the rewrite rules acting on a \MELLdaimon quasi-proof-structures $R$, we will only draw the parts of 
$R$ belonging to the relevant component; \eg, if we are interested in
an $\ax$-cell 
whose outputs are the conclusions $i$ and $i\!+\!1$, 
and it is the only cell in
a component, we will write
\parbox[b][2\baselineskip]{1.5cm}{\(
  \scalebox{\magicnumber}{
  \tikzsetnextfilename{images/mell-ax-example}
  \pnet{
  \pnformulae{
  \pnf[i]{$i$}~\pnf[j]{$i\!+\!1\,$}
  }
  \pnaxiom[ax]{i,j}
  \pnsemibox{ax}
  }}
\)}
ignoring the rest.

\begin{figure}[t]
  \centering
  \vspace{\beforepn}
  \begin{subfigure}[c]{0.2\textwidth}
  	\begin{align*}
  	\scalebox{\magicnumber}{
  		\tikzsetnextfilename{images/mell-exc-rewrite-1}
  		\pnet{
  			\pnsomenet[u]{}{1.75cm}{0.75cm}[at (-1,0.75)]
  			\pnoutfrom{u.-150}[A]{$\Gamma_k\,$}
  			\pnoutfrom{u.-130}[B]{$i$}
  			\pnoutfrom{u.-70}[A']{$i\!+\!1$}
  			\pnoutfrom{u.-30}[B']{$\,\,\Gamma_{k}'$}
  			\pnsemibox{u,A,A',B,B'}
  	}} \pnrewrite[\Ex_{i}]
  	\tikzsetnextfilename{images/mell-exc-rewrite-2}
  	\scalebox{\magicnumber}{
  		\pnet{
  			\pnsomenet[u]{}{1.75cm}{0.75cm}[at (-1,0.75)]
  			\pnoutfrom{u.-150}[A]{$\Gamma_{k}\,$}
  			\pnoutfrom{u.-110}[B]{$i\!+\!1$}
  			\pnoutfrom{u.-50}[A']{$i$}
  			\pnoutfrom{u.-30}[B']{$\,\,\Gamma_{k}'$}
  			\pnsemibox{u,A,B,A',B'}
  		}
  	}
  	\end{align*}
  			\vspace{-\baselineskip}
  	\caption{Exchange}
  	\label{fig:exc}
  \end{subfigure}%
  \begin{subfigure}[c]{0.2\textwidth}
    \begin{align*}
      \scalebox{\magicnumber}{
      \tikzsetnextfilename{images/mell-mix-rewrite-1}
      \pnet{
      \pnsomenet[u]{}{1cm}{0.75cm}[at (-1,0.75)]
      \pnoutfrom{u.-125}[A]{$\Gamma_k\,$}
      \pnoutfrom{u.-55}[B]{$i$}
      \pnsomenet[u']{}{1cm}{0.75cm}[at (0.3,0.75)]
      \pnoutfrom{u'.-125}[A']{\small$i\!+\!1$}
      \pnoutfrom{u'.-55}[B']{$\,\,\Gamma_k'$}
      \pnsemibox{u,u',A,A',B,B'}
      }} \pnrewrite[\Mix_{i}]
      \tikzsetnextfilename{images/mell-mix-rewrite-2}
      \scalebox{\magicnumber}{
      \pnet{
      \pnsomenet[u]{}{1cm}{0.75cm}[at (-1,0.75)]
      \pnoutfrom{u.-125}[A]{$\Gamma_{k}$}
      \pnoutfrom{u.-55}[B]{$i$}
      \pnsomenet[u']{}{1cm}{0.75cm}[at (0.8,0.75)]
      \pnoutfrom{u'.-125}[A']{\small$i\!+\!1$}
      \pnoutfrom{u'.-55}[B']{$\,\,\Gamma_k'$}
      \pnsemibox{u,A,B}
      \pnsemibox{u',A',B'}
      }
      }
    \end{align*}
    		\vspace{-\baselineskip}
    \caption{Mix}
    \label{fig:mix}
  \end{subfigure}%
  \begin{subfigure}[c]{0.3\textwidth}
    \begin{align*}
      \scalebox{\magicnumber}{
      \tikzsetnextfilename{images/mell-ax-rewrite-1}
      \pnet{
      \pnformulae{
      \pnf{$\cdots\,$}~\pnf[i]{$i$}~\pnf[j]{$i\!+\!1$}~\pnf{$\,\cdots$}
      }
      \pnaxiom[ax]{i,j}
      \pnsemibox{ax,i,j}
      }} \pnrewrite[\Ax_{i}]
      \tikzsetnextfilename{images/mell-ax-rewrite-2}
      \scalebox{\magicnumber}{
      \pnet{
      \pnformulae{
      \pnf{$\cdots$}~\pnf{$\cdots$}
      }      
      }}
    \end{align*}
    		\vspace{-\baselineskip}
    \caption{Hypothesis (\(\Ax, \maltese, \One, \bot, \Weak\))}
    \label{fig:hypothesis}
  \end{subfigure}
  \begin{subfigure}[b]{0.1\textwidth}
    \begin{align*}
      \scalebox{\magicnumber}{
      \tikzsetnextfilename{images/mell-cutaxiom-1}
      \pnet{
      \pnsomenet[u]{}{1.5cm}{0.75cm}[at (-2,0)]
      \pnoutfrom{u.-40}[A]{$\quad$}
      \pnoutfrom{u.-95}[B]{$\quad$}
      \pnoutfrom{u.-140}[k]{$\Gamma_k$}
      \pncut[cut]{A,B}
      \pnsemibox{u,cut}
      }}
      \pnrewrite[\Cut^i]
      \tikzsetnextfilename{images/mell-cutaxiom-2}
      \scalebox{\magicnumber}{
      \pnet{
      \pnsomenet[u]{}{1.5cm}{0.75cm}[at (-2,0)]
      \pnoutfrom{u.-40}[B]{$i\!+\!1$}
      \pnoutfrom{u.-95}[A]{$i$}
      \pnoutfrom{u.-140}[k]{$\Gamma_k$}
      \pnsemibox{u,A,B}
      }
      }
    \end{align*}
    		\vspace{-\baselineskip}
    \caption{Cut}
    \label{fig:cut}
  \end{subfigure}%
  \begin{subfigure}[b]{0.3\textwidth}
    \begin{align*}
      \tikzsetnextfilename{images/mell-tensor-rewrite-1}
      \scalebox{\magicnumber}{
      \pnet{
      \pnsomenet[u]{}{1.5cm}{0.75cm}
      \pnoutfrom{u.-140}[k]{$\Gamma_{k}$}
      \pnoutfrom{u.-95}[A]{$\quad$}
      \pnoutfrom{u.-40}[B]{$\quad$}
      \pntensor{A,B}[i]{$i$}
      \pnsemibox{u,k,A,B,i}
      }}
      \pnrewrite[\otimes_i]
      \tikzsetnextfilename{images/mell-tensor-rewrite-2}
      \scalebox{\magicnumber}{
      \pnet{
      \pnsomenet[u]{}{1.5cm}{0.75cm}
      \pnoutfrom{u.-140}[k]{$\Gamma_{k}$}
      \pnoutfrom{u.-95}[A]{$i$}
      \pnoutfrom{u.-40}[B]{$i\!+\!1$}
      \pnsemibox{u,k,A,B}
      }}
    \end{align*}
    		\vspace{-\baselineskip}
    \caption{Binary multiplicative (\(\otimes, \parr\))}
    \label{fig:tensor}
  \end{subfigure}
  \begin{subfigure}[b]{0.2\textwidth}
    \begin{align*}
      \tikzsetnextfilename{images/mell-contraction-rewrite-1}
      \scalebox{\magicnumber}{
      \pnet{
      \pnsomenet[u]{}{2cm}{0.75cm}
      \pnoutfrom{u.-155}[k]{$\Gamma_{k}$}
      \pnoutfrom{u.-125}[C]{\small$\cdots$}
      \pnoutfrom{u.-80}[B]{$\quad$}
      \pnoutfrom{u.-40}[A]{$\quad$}
      \pnoutfrom{u.-25}[D]{\small$\,\cdots$}
       \pnexp{}{A,B,C, D}[i]{$i$}
      \pnsemibox{u,i}
      }}
      \pnrewrite[\Contr_i]
      \tikzsetnextfilename{images/mell-contraction-rewrite-2}
      \scalebox{\magicnumber}{
      \pnet{
      \pnsomenet[u]{}{2cm}{0.75cm}
      \pnoutfrom{u.-155}[k]{$\Gamma_{k}$}
      \pnoutfrom{u.-125}[C]{\small$\cdots$}
      \pnoutfrom{u.-80}[B]{$\quad$}
      \pnexp{}{C,B}[i]{$i$}
      \pnoutfrom{u.-40}[A]{$\quad$}
      \pnoutfrom{u.-25}[D]{\small$\,\cdots$}
      \pnexp{}{A,D}[i1]{$i\!+\!1$}
      \pnsemibox{u,A,i,i1}
      }}
    \end{align*}
    		\vspace{-\baselineskip}
    \caption{Contraction}
    \label{fig:contraction}
  \end{subfigure}%
  \begin{subfigure}[b]{0.2\textwidth}
    \begin{align*}
      \tikzsetnextfilename{images/mell-dereliction-rewrite-1}
      \scalebox{\magicnumber}{
      \pnet{
      \pnsomenet[u]{}{1.25cm}{0.75cm}
      \pnoutfrom{u.-130}[k]{$\Gamma_k$}
      \pnoutfrom{u.-50}[A]{$\quad$}
      \pnexp{}{A}[i]{$i$}
      \pnsemibox{u,k,A,i}
      }}
      \pnrewrite[\Der_i]
      \tikzsetnextfilename{images/mell-dereliction-rewrite-2}
      \scalebox{\magicnumber}{
      \pnet{
      \pnsomenet[u]{}{1.25cm}{0.75cm}
      \pnoutfrom{u.-130}[k]{$\Gamma_{k}$}
      \pnoutfrom{u.-50}[A]{$i$}
      \pnsemibox{u,k,A}
      }}
    \end{align*}
    		\vspace{-\baselineskip}
    \caption{Dereliction}
    \label{fig:dereliction}
  \end{subfigure}%
	\begin{subfigure}[b]{0.2\textwidth}
		 \begin{align*}
			\tikzsetnextfilename{images/mell-box-1}
			\scalebox{\magicnumber}{
				\pnet{
					\pnsomenet[u]{}{1.25cm}{0.75cm}
					\pnbox{u}
					\pnoutfrom{u.-50}[pi1c]{$ $}
					\pnbag{}{pi1c}{$i$}
					\pnoutfrom{u.-130}[u2a2]{$ $}
					\pnexp{}{u2a2}[?u2a2]{\small $\wn \Gamma_{k}$}    
					\pnsemibox{pi1c,u2a2,?u2a2,u}
			}}
			\pnrewrite[\BoxR_{i}]
			\tikzsetnextfilename{images/mell-box-2}
			\scalebox{\magicnumber}{
				\pnet{
					\pnsomenet[u]{}{1.25cm}{0.75cm}
					\pnoutfrom{u.-50}[pi1c]{$i$}[3.55cm][]
					\pnoutfrom{u.-130}[u2a2]{$\quad$}
					\pnexp{}{u2a2}[?u2a2]{\small$\wn \Gamma_{k}$}
					\pnsemibox{u,pi1c,u2a2,?u2a2}
			}}
		\end{align*}
				\vspace{-\baselineskip}
		\caption{Box}
		\label{fig:box}
	\end{subfigure}
  \caption{Actions of elementary paths on $\MELLdaimon$ quasi-proof-structures.}
  \label{fig:actions-all}
\end{figure}

\begin{definition}[action of paths on $\MELL$ quasi-proof-structures]
  \label{def:unwinding}
  An elementary path {$a \colon \Gamma \to \Gamma'$} defines a relation
  $\pnrewrite[a] \ \subseteq \qMELL(\Gamma) \times \qMELL(\Gamma')$ (the \emph{action} of $a$)
  as the smallest relation containing all the cases in \Cref{fig:actions-all},
  with the following remarks:
  \begin{description}
  \item[mix] read in reverse, a quasi-proof-structure with two components is in
    relation with a proof-structure with the same module but the two roots of
    such components merged.
  \item[hypothesis] if
    $a\in \{\Ax_{i}, \maltese_{i}, \One_i, \bot_i, \Weak_i \}$, the rules
    have all in common to act by deleting a cell without inputs that is the only
    cell in its component. We have drawn the axiom case in
    \Cref{fig:hypothesis}, the others vary only by their number of conclusions.
  \item[cut] read in reverse, a quasi-proof-structure with two conclusions $i$
    and $i+1$ is in relation with the quasi-proof-structure where these two
    conclusions are cut. 
    This rule, from left to right, is non-deterministic (as
    there are many possible cuts).
  \item[binary multiplicatives] these rules delete a binary connective. We have
    only drawn the $\otimes$ case in \Cref{fig:tensor}, the $\parr$ case is
    similar.  
  \item[contraction] 
  splits 
  a $\wn$-cell with $h\!+\!k\!+\!2$ inputs into two $\wn$-cells with $h\!+\!1$ and $k\!+\!1$ inputs, respectively.
  \item[dereliction] 
  only applies if the $\wn$-cell (with $1$ input) does not shift a level in the box-forest.
  \item[box] 
  only applies if a box (and its frontier) is alone in its component.
  \end{description}

  This definition of the rewrite system is extended to define a relation
  $\pnrewrite[\xi] \ \subseteq \qMELL(\Gamma) \times \qMELL(\Gamma')$ (the \emph{action} of any path $\xi \colon \Gamma \to \Gamma'$) by composition of relations.
\end{definition}


Given two  \MELLdaimon quasi-proof-structures $R$ and $R'$, we say that a rule $a$ \emph{applies} to $R$ if there is a finite sequence of exchanges $\Ex_{i_1}\dots\Ex_{i_n}$ such that $R\pnrewrite[\Ex_{i_1}\dots\Ex_{i_n}a]R'$.

\begin{definition}[the functor $\qMELL$]
\label{def:FunctorqMELL}
  We define a functor $\qMELL \colon \Path \to \Rel$ by:
  \begin{itemize}
  \item \emph{on objects:} $\qMELL(\Gamma)$ is the set of $\MELLdaimon$ quasi-proof-structures of
    type $\Gamma$;
  \item \emph{on morphisms:} for $\xi \colon \Gamma \to \Gamma'$, $\qMELL(\xi) = \ \pnrewrite[\xi]$ (see \Cref{def:unwinding}).
  \end{itemize}
\end{definition}

  
Our rewrite rules enjoy two useful properties, expressed by Propositions \ref{prop:unwinding-co-functional} and \ref{prop:unwinding-to-empty}.

\begin{proposition}[co-functionality]
  \label{prop:unwinding-co-functional}
  Let $\xi \colon \Gamma \to \Gamma'$ be a path.
  The relation $\pnrewrite[\xi]$ is a co-function on the sets of underlying
  graphs, that is, 
	a function $ \op{\pnrewrite[\xi]} \colon \qMELL(\Gamma') \to \qMELL(\Gamma)$.
\end{proposition}




\begin{lemma}[applicability of rules]
  \label{lem:cell-or-mix}
  Let $R$ be a non-empty $\MELLdaimon$ quasi-proof-structure.
  There exists a conclusion $i$ such that:
  \begin{itemize}
  \item either a rule in
    $\{\Ax_i, \One_i, \bot_i, \otimes_i, \parr_i, \Contr_i, \Der_i, \Weak_i,
    \Cut^i, \maltese_{i}, \allowbreak \BoxR_{i}\}$ applies to $R$;
		\item or $R \pnrewrite[\Mix_i] R'$ (where the conclusions affected by $\Mix_{i}$ are $i\!-\!k, \dots, i, i\!+\!1, \dots, i\!+\!\ell$) and $i\!-\!k,\dots,i$ are all the conclusions of either a box or an
    hypothesis cell, 
    and one of the components of $R'$ coincides with this cell or box (and~its~border).
  \end{itemize}
\end{lemma}

\noindent \Cref{prop:unwinding-co-functional} and \Cref{lem:cell-or-mix} are proven by simple inspection of the rewrite rules of \Cref{fig:actions-all}.

\begin{proposition}[termination]
  \label{prop:unwinding-to-empty}
  Let $R$ be a $\MELLdaimon$ quasi-proof-structure of type $\Gamma$.
  There exists a path $\xi \colon \Gamma \to \emptynet$ such that
  $R \pnrewrite[\xi] \emptynet$.
\end{proposition}

To prove \Cref{prop:unwinding-to-empty}, it is enough to apply \Cref{lem:cell-or-mix} and show that the size of \MELLdaimon quasi-proof-structures decreases for each application of the rules in \Cref{fig:actions-all}, according to the following definition of size.
The \emph{size} of a proof-structure $R$ is the couple $(p,q)$ where
	\begin{itemize}
		\item $p$ is the (finite) multiset of the number of inputs of each $\wn$-cell in $R$;
		\item $q$ is the number of cells not labeled by $\maltese$ in $R$.
	\end{itemize}
	The \emph{size} of a quasi-proof-structure $R$ is the (finite) multiset of the sizes of its components. 
	Multisets are ordered as usual, couples are ordered lexicographically.


%% file: naturality.tex
\section{Naturality of unwinding \texorpdfstring{\polyadicdaimon}{DiLL0}
  quasi-proof-structures}
\label{sec:naturality}

For $\Gamma$ a list of lists of $\MELL$ formul\ae{}, $\PolyPN(\Gamma)$ is the
set of \polyadicdaimon quasi-proof-structures of type $\Gamma$. 
For any set $X$, its powerset is denoted by $\PowerSet{X}$.

\begin{figure}[!t]
	\centering
	\vspace*{\beforepn}
	\begin{subfigure}[c]{0.3\textwidth}
		\begin{align*}
		\tikzsetnextfilename{images/mix-daimon-rewrite-1}
		\scalebox{\magicnumber}{
			\pnet{
				\pnformulae{
					\pnf[i]{$\Maybe{\Gamma_{k}}$}
					~\pnf[l]{$\Maybe{i}$}~\pnf[l1]{$\Maybe{i\!+\!1}$}~\pnf[j]{$\Maybe{\Gamma_{k}'}  $}
				}
				\pndaimon[d]{i,l,l1,j}
				\pnsemibox{d,i,l,l1,j}
		}} \pnrewrite[\Mix_{i}]
		\tikzsetnextfilename{images/mix-daimon-rewrite-2}
		\Bigg\{
		\scalebox{\magicnumber}{
			\pnet{
				\pnformulae{
					\pnf[i]{$\Maybe{\Gamma_{k}}$}~\pnf[ik]{$\Maybe{i}$}
					~~\pnf[j1]{$\Maybe{i\!+\!1}$}~\pnf[j]{$\Maybe{\Gamma_{k}'}$}
				}
				\pndaimon[d1]{i,ik}
				\pndaimon[d2]{j1,j}
				\pnsemibox{d1,i,ik}
				\pnsemibox{d2,j1,j}
		}}
		\Bigg\}
		\end{align*}
				\vspace{-\baselineskip}
		\caption{Mix}
		\label{fig:mix-daimon}
	\end{subfigure}%
	\begin{subfigure}[c]{0.3\textwidth}
		\begin{align*}
		\tikzsetnextfilename{images/daimon-ax-rewrite-1}
		\scalebox{\magicnumber}{
			\pnet{
				\pnformulae{
					\pnf{$\dots\,$}~\pnf[i]{$i$}~\pnf[j]{$i\!+\!1$}~\pnf{$\,\dots$}
				}
				\pndaimon[ax]{i,j}
				\pnsemibox{ax,i,j}
		}} \pnrewrite[\Ax_{i}]
		\tikzsetnextfilename{images/daimon-ax-rewrite-2}
		\Bigg\{
		\scalebox{\magicnumber}{
			\pnet{
				\pnformulae{
					\pnf{$\dots$}~\pnf{$\dots$}
				}
		}}
		\Bigg\}
		\end{align*}
				\vspace{-\baselineskip}
		\caption{Hypothesis \texorpdfstring{(\(\Ax, \maltese, \One, \bot, \Weak\))}{}}
		\label{fig:ax-daimon}
	\end{subfigure}

\vspace{-1.5\baselineskip}
	\begin{subfigure}[c]{0.3\textwidth}
		\begin{align*}
		\tikzsetnextfilename{images/daimon-cut-rewrite-1}
		\scalebox{\magicnumber}{
			\pnet{
				\pnformulae{
					\pnf[i]{$\Maybe{\Gamma_k}$}~\pnf[l]{ }
				}
				\pndaimon[d]{i}[1][.5]
				\pnsemibox{d,i}
		}} \pnrewrite[\Cut^i]
		\tikzsetnextfilename{images/daimon-cut-rewrite-2}
		\Bigg\{
		\scalebox{\magicnumber}{
			\pnet{
				\pnformulae{
					\pnf[i]{$\Maybe{\Gamma_{k}}$}~\pnf[n]{$i$}~\pnf[n']{$i\!+\!1$}~
				}
				\pndaimon[d]{i,n,n'}
				\pnsemibox{d,i,n,n'}
		}}
		\Bigg\}
		\end{align*}
				\vspace{-\baselineskip}
		\caption{Cut}
		\label{fig:cut-daimon}
	\end{subfigure}%
	\begin{subfigure}[c]{0.3\textwidth}
		\begin{align*}
		\tikzsetnextfilename{images/daimon-tensor-rewrite-1}
		\scalebox{\magicnumber}{
			\pnet{
				\pnformulae{
					\pnf[i]{$\Maybe{\Gamma_k}$}~~\pnf[l]{$i$}
				}
				\pndaimon[d]{i,l}[1][.2]
				\pnsemibox{d,i}
		}} \pnrewrite[\Contr_i]
		\tikzsetnextfilename{images/daimon-tensor-rewrite-2}
		\Bigg\{
		\scalebox{\magicnumber}{
			\pnet{
				\pnformulae{
					\pnf[i]{$\Maybe{\Gamma_{k}}$}~\pnf[n]{$i$}~\pnf[n']{$i\!+\!1$}~
				}
				\pndaimon[d]{i,n,n'}
				\pnsemibox{d,i,n,n'}
		}}
		\Bigg\}
		\end{align*}
				\vspace{-\baselineskip}
		\caption{Binary rule \texorpdfstring{(\(\otimes, \parr, \Contr\))}{}}
		\label{fig:daimon-tensor}
	\end{subfigure}%
	\begin{subfigure}[c]{0.2\textwidth}
		\begin{align*}
		\tikzsetnextfilename{images/daimon-dereliction-rewrite-1}
		\scalebox{\magicnumber}{
			\pnet{
				\pnformulae{
					\pnf[i]{\small $\Maybe{\Gamma_k}$}~\pnf[l]{$i$}
				}
				\pndaimon[d]{i,l}
				\pnsemibox{d,i,l}
		}} \pnrewrite[\Der_i]
		\tikzsetnextfilename{images/daimon-dereliction-rewrite-2}
		\Bigg\{
		\scalebox{\magicnumber}{
			\pnet{
				\pnformulae{
					\pnf[i]{\small $\Maybe{\Gamma_{k}}$}~\pnf[n']{\small $i$}~
				}
				\pndaimon[d]{i,n'}
				\pnsemibox{d,i,n'}
		}}
		\Bigg\}
		\end{align*}
				\vspace{-\baselineskip}
		\caption{Dereliction}
		\label{fig:daimon-dereliction}
	\end{subfigure}
	\begin{subfigure}[b]{0.2\textwidth}
		\tikzsetnextfilename{images/daimon-box-rewrite-1}
		\begin{align*}
		\scalebox{\magicnumber}{
			\pnet{
				\pnformulae{
					\pnf[i1]{$\wn{\Gamma_k}$}~\pnf[i]{$i$}
				}
				\pndaimon[dai]{i1,i}
				\pnsemibox{dai,i1,i}}
		} \pnrewrite[\BoxR_i]
		\tikzsetnextfilename{images/daimon-box-rewrite-1}
		\Bigg\{
		\scalebox{\magicnumber}{
			\pnet{
				\pnformulae{
					\pnf[i1]{$\wn{\Gamma_k}$}~\pnf[i]{$i$}
				}
				\pndaimon[dai]{i1,i}
				\pnsemibox{dai,i1,i}}
		}
		\Bigg\}
		\end{align*}
				\vspace{-\baselineskip}
		\caption{Daimoned box}
		\label{fig:daimon-box}
	\end{subfigure}%
	\begin{subfigure}[b]{0.2\textwidth}
		\tikzsetnextfilename{images/empty-box-rewrite-1}
		\begin{align*}
		\scalebox{\magicnumber}{
			\pnet{
				\pnformulae{
					\pnf[Gamma2]{$\wn\Gamma_k$}~
					\pnf[i]{$i$}
				}
				\pnwn[?2]{Gamma2}
				\pncown[!i]{i}
				\pnsemibox{Gamma2,i,?2,!i}
		}}
		\pnrewrite[\BoxR_i]
		\tikzsetnextfilename{images/empty-box-rewrite-2}
		\Biggl\{
		\scalebox{\magicnumber}{
			\pnet{
				\pnformulae{
					\pnf[Gamma2]{$\wn\Gamma_k$}~
					\pnf[i]{$i$}
				}
				\pndaimon[dai]{i,Gamma2}
				\pnsemibox{dai,i,Gamma2}
			}
		}
		\Biggr\}
		\end{align*}
				\vspace{-\baselineskip}
		\caption{Empty box}
		\label{fig:empty-box}   
	\end{subfigure}%
	\begin{subfigure}[b]{0.4\textwidth}
		\begin{align*}
		\tikzsetnextfilename{images/oc-split-1}
		\scalebox{\magicnumber}{
			\pnet{
				\pnformulae{
					\pnf{$\dots$}}
				\pnsomenet[pin]{$\rho_n$}{1.5cm}{0.5cm}[at (0,1.5)]
				\pnsomenet[pi1]{$\rho_1$}{1.5cm}{0.5cm}[at (0,0)]
				\pnoutfrom{pi1.-45}[pi1c]{$\quad$}
				\pnoutfrom{pin.-30}[pinc]{$\quad$}
				\pnbag[!]{}{pi1c,pinc}{$i$}[1][0.5]
				\pnoutfrom{pi1.-135}[pi1a1]{$\quad$}
				\pnoutfrom{pin.-150}[pina2]{$\quad$}
				\pnexp[?1]{}{pi1a1,pina2}{$\wn \Gamma_{k}$}[1][-0.5]				
				\pnsemibox{pi1,pin,!,?1}
		}}
		\pnrewrite[\BoxR_i]
		\tikzsetnextfilename{images/oc-split-2}
		\left\{
		\scalebox{\magicnumber}{
			\pnet{
				\pnformulae{
					\pnf{$\dots$}}
				\pnsomenet[pi1]{$\rho_j$}{1.5cm}{0.5cm}[at (0,0)]
				\pnoutfrom{pi1.-30}[pi1c]{$i$}[2]
				\pnoutfrom{pi1.-150}[pi1a1]{$\quad$}
				\pnexp[?2]{}{pi1a1}{$\wn \Gamma_{k}$}[1][-0.5]
				\pnsemibox{pi1,?2}
		}}\right\}_{1 \leqslant j \leqslant n}
		\end{align*}
		\vspace{-\baselineskip}
		\caption{Non-empty box ($n > 0$)}
		\label{fig:polybox}
	\end{subfigure}
	\caption{Actions of elementary paths on \texorpdfstring{$\maltese$}{}-cells and on a
      box in \texorpdfstring{$\PolyPN$}{DiLL}.}
	\label{fig:actions-daimon}
\end{figure}

\begin{definition}[action of paths on \texorpdfstring{\polyadicdaimon}{DiLL}
  quasi-proof-structures]
	\label{def:paths-on-polyadic}
	An elementary path $a \colon \Gamma \to \Gamma'$ defines a relation
	$\pnrewrite[a] \  \subseteq \PolyPN(\Gamma) \times \PowerSet{\PolyPN(\Gamma')}$ (the \emph{action} of $a$) by the
	rules in \Cref{fig:actions-all} (except \Cref{fig:box}, and with all the already remarked notes) and in
	\Cref{fig:actions-daimon}. 
	
	We extend this relation on
	$\PowerSet{\PolyPN(\Gamma)} \times\PowerSet{\PolyPN(\Gamma')}$ by the monad
	multiplication of $X \mapsto \PowerSet{X}$ and define $\pnrewrite[\xi]$ (the \emph{action} of any path $\xi \colon \Gamma \to \Gamma'$) by composition of relations.
\end{definition}

Roughly, all the rewrite rules in \Cref{fig:actions-daimon}---except
\Cref{fig:polybox}---mimic the behavior of the corresponding rule in \Cref{fig:actions-all} using a $\maltese$-cell. 
Note that in \Cref{fig:empty-box} a $\maltese$-cell is created.

The non-empty box rule 
in \Cref{fig:polybox} requires that, 
	on the left of $\pnrewrite[\BoxR_i]$, $\rho_{j}$ is not connected to $\rho_{j'}$ for $j \neq j'$, except for the $\oc$-cell and the $\wn$-cells in the conclusions.
Read in reverse, the rule
associates with a non-empty finite set of \polyadic quasi-proof-structures
$\{\rho_1, \dots, \rho_n\}$ the merging of $\rho_1, \dots, \rho_n$, that is the
\polyadic quasi-proof-structure depicted 
on the left of $\pnrewrite[\BoxR_i]$.


\begin{definition}[the functor \texorpdfstring{$\PPolyPN$}{PPolyPN}]
	\label{def:FunctorPPoly}
	We define a functor $\PPolyPN \colon \Path \to \Rel$ by:
	\begin{itemize}
		\item\emph{on objects:} for $\Gamma$ a list of lists of \MELL formul\ae{}, $\PPolyPN(\Gamma) = \PowerSet{\PolyPN(\Gamma)}$, the set of sets of \polyadicdaimon proof-structures of type $\Gamma$;
		\item\emph{on morphisms:} for $\xi : \Gamma \to \Gamma'$, $\PPolyPN(\xi) = \ \pnrewrite[\xi]$ (see \Cref{def:paths-on-polyadic}).
	\end{itemize}
\end{definition}

\begin{theorem}[naturality]
  \label{thm:projection-natural}
  The \filled Taylor expansion
  \SLV{}{\marginpar{\scriptsize Proof in \\\Cref{sec:naturality-proof}, p.~\pageref{sec:naturality-proof}}}
  defines a natural transformation
 	$\NatTransf \colon \PPolyPN \!\Rightarrow\! \qMELL \colon \Path \!\to\! \Rel$
  by: 
	$(\Pi, R) \!\in\! \NatTransf_{\Gamma}$ iff $\Pi \!\subseteq\! \FatTaylor{R}$ and the type of $R$ is $\Gamma$.
	Moreover, if $\Pi$ is a set of \polyadic proof-structures with $\Pi \pnrewrite[\xi] \Pi'$ and $\Pi' \subseteq \Taylor{R'}$, then $R$ is a \MELL proof-structure and $\Pi \subseteq \Taylor{R}$, where $R$ is such that $R \pnrewrite[\xi] R'$.\footnotemark
	\footnotetext{The part of the statement after ``moreover'' is our way to control the presence of $\maltese$-cells.}
\end{theorem}
In other words, the following diagram commutes for every path
$\xi \colon \Gamma \to \Gamma'$.
\begin{commutativediagram}{3}{5}{5}
  {
    \PPolyPN(\Gamma) \& \PPolyPN(\Gamma')\\
    \qMELL(\Gamma) \& \qMELL(\Gamma')\\
  };
  \path[-stealth]
  (m-1-1) edge [auto] node {$\PPolyPN(\xi)$} (m-1-2)
  (m-1-1) edge [auto] node {$\NatTransf_{\Gamma}$} (m-2-1)
  (m-2-1) edge [auto] node {$\qMELL(\xi)$} (m-2-2)
  (m-1-2) edge [auto] node {$\NatTransf_{\Gamma'}$} (m-2-2);
\end{commutativediagram}
It means that given $\Pi \pnrewrite[\xi] \Pi'$, where $\Pi'\subseteq \FatTaylor{R'}$,
we can simulate backwards the rewriting to $R$ (this is where the
co-functionality of the rewriting steps expressed by \Cref{prop:unwinding-co-functional} comes handy) so that $R \pnrewrite[\xi] R'$ and $\Pi\subseteq \FatTaylor{R}$; and conversely, given
$R \pnrewrite[\xi] R'$, we can simulate the rewriting for any $\Pi\subseteq \FatTaylor{R}$, so that $\Pi \pnrewrite[\xi] \Pi'$ for some $\Pi'\subseteq \FatTaylor{R'}$.

%% file: glueability.tex
\section{Glueability of \texorpdfstring{\polyadic}{DiLL} quasi-proof-structures}
\label{sec:glueable}

Naturality (\Cref{thm:projection-natural}) allows us to characterize the sets of \polyadic proof-structures that are in the Taylor expansion of some $\MELL$ proof-structure (\Cref{thm:characterization} below).

\begin{definition}[glueability]
  We say that a set $\Pi$ of \polyadicdaimon quasi-proof-structures is \emph{glueable}, if
  there exists a path $\xi$ such that
  $ \Pi \pnrewrite[\xi] \{\emptynet\}$.
\end{definition}

\begin{theorem}[glueability criterion]
  \label{thm:characterization}
%
Let \(\Pi\) be a set of \polyadic proof-structures:
\(\Pi\) is glueable if and only if $\Pi \subseteq \Taylor{R}$ for some \(\MELL\) proof-structure $R$.
%

\end{theorem}

\begin{proof}
	If $\Pi \subseteq \Taylor{R}$ for some $\MELL$ proof-structure $R$,
	then by termination (\Cref{prop:unwinding-to-empty})
	$R \pnrewrite[\xi] \emptynet$ for some path $\xi$, and so
	$\Pi \pnrewrite[\xi] \{\emptynet\}$ by naturality
	(\Cref{thm:projection-natural}, as 	$\FatTaylor{\emptynet} = \{\emptynet\}$).
	
  Conversely, if $\Pi \pnrewrite[\xi] \{\emptynet\}$ for some path $\xi$, then
  by naturality (\Cref{thm:projection-natural}, as
  $\Taylor{\emptynet} = \{\emptynet\}$ and $\Pi$ is a set of \polyadic proof-structures)
  $\Pi \subseteq \Taylor{R}$	\mbox{for some $\MELL$ proof-structure $R$.}
\end{proof}


\begin{example}
	\label{ex:not-coherent}
  The three \polyadic proof-structures $\rho_1, \rho_2, \rho_3$ below 
  are not glueable as a whole, but are glueable two by two. 
  In fact, there is no \MELL proof-structure whose Taylor expansion contains $\rho_1, \rho_2, \rho_3$, but any pair of them is in the Taylor expansion of some \MELL proof-structure. 
%
  This is a slight variant of the example in \cite[pp.~244-246]{Tasson:2009}.
  
  \begin{center}
  \vspace{\beforepn}
  \scalebox{0.6}{
  	\pnet{
  		\pnformulae{
  			\pnf[1]{$\mathbf{1}$}~\pnf[2]{$\mathbf{1}$}~\pnf[3]{$\mathbf{1}$}~
  			\pnf[4]{$\mathbf{1}$}~
  			\pnf[5]{$\bot$}~\pnf[6]{$\bot$}~\pnf[7]{$\bot$}~
  			\pnf[8]{$\bot$}~\pnf[9]{$\bot$}
  		}
  		\pnone{1}
  		\pnone{2}
  		\pnone{3}
  		\pnone{4}
  		\pnbot{5}
  		\pnbot{6}
  		\pnbot{7}
  		\pnbot{8}
  		\pnbot{9}
  		\pnbag{}{1,2}{$\oc \mathbf{1}$}
  		\pnbag{}{3}{$\oc \mathbf{1}$}
  		\pnbag{}{4}{$\oc \mathbf{1}$}
  		\pnexp{}{5}{$\wn \bot$}
  		\pnexp{}{6,7}{$\wn \bot$}
  		\pnexp{}{8,9}{$\wn \bot$}
  }}
  \qquad \ 
  \tikzsetnextfilename{images/tasson-example-2}
  \scalebox{0.6}{
  	\pnet{
  		\pnformulae{
  			\pnf[1]{$\mathbf{1}$}~\pnf[2]{$\mathbf{1}$}~\pnf[3]{$\mathbf{1}$}~
  			\pnf[4]{$\mathbf{1}$}~
  			\pnf[5]{$\bot$}~\pnf[6]{$\bot$}~\pnf[7]{$\bot$}~
  			\pnf[8]{$\bot$}~\pnf[9]{$\bot$}
  		}
  		\pnone{1}
  		\pnone{2}
  		\pnone{3}
  		\pnone{4}
  		\pnbot{5}
  		\pnbot{6}
  		\pnbot{7}
  		\pnbot{8}
  		\pnbot{9}
  		\pnbag{}{1}{$\oc \mathbf{1}$}
  		\pnbag{}{2,3}{$\oc \mathbf{1}$}
  		\pnbag{}{4}{$\oc \mathbf{1}$}
  		\pnexp{}{5,6}{$\wn \bot$}
  		\pnexp{}{7}{$\wn \bot$}
  		\pnexp{}{8,9}{$\wn \bot$}
  }}
  \qquad \ 
  \tikzsetnextfilename{images/tasson-example-3}
  \scalebox{0.6}{
  	\pnet{
  		\pnformulae{
  			\pnf[1]{$\mathbf{1}$}~\pnf[2]{$\mathbf{1}$}~\pnf[3]{$\mathbf{1}$}~
  			\pnf[4]{$\mathbf{1}$}~
  			\pnf[5]{$\bot$}~\pnf[6]{$\bot$}~\pnf[7]{$\bot$}~
  			\pnf[8]{$\bot$}~\pnf[9]{$\bot$}
  		}
  		\pnone{1}
  		\pnone{2}
  		\pnone{3}
  		\pnone{4}
  		\pnbot{5}
  		\pnbot{6}
  		\pnbot{7}
  		\pnbot{8}
  		\pnbot{9}
  		\pnbag{}{1}{$\oc \mathbf{1}$}
  		\pnbag{}{2}{$\oc \mathbf{1}$}
  		\pnbag{}{3,4}{$\oc \mathbf{1}$}
  		\pnexp{}{5,6}{$\wn \bot$}
  		\pnexp{}{7,8}{$\wn \bot$}
  		\pnexp{}{9}{$\wn \bot$}
      }
    }
  \end{center}
\end{example}

An example of the action of a path starting from a \polyadic
proof-structure $\rho$ and ending in \(\{\emptynet\}\) can be found in
\Cref{fig:longex1,fig:longex2}. Note that it is by no means the shortest
possible path. When replayed backwards, it induces a \(\MELL\) proof-structure
\(R\) such that \(\rho\in \Taylor{R}\).

\begin{figure}[t]
  \centering
  \vspace{\beforepn}
  \begin{align*}
    \begin{array}{cccccc}
      \rho=
      &\tikzsetnextfilename{images/long-path-poly-1}
        \Biggl\{
        \scalebox{\smallproofnets}
        {\pnet{
        \pnformulae{
        ~\pnf[?b]{\small $\wn\wn \bot$}~~~
        \pnf[!A-oA]{\small$\!\!\!\oc\oc(A^{\bot} \!\parr\! A) \quad $}~\\
        }
        \pnwn[wn]{?b}
        \pncown[cown]{!A-oA}
        \pnsemibox{?b,!A-oA,wn,cown}
      }}
      \Biggr\}&\pnrewrite[\BoxR_2]&
      \tikzsetnextfilename{images/long-path-poly-2}
      \Biggl\{
      \scalebox{\smallproofnets}                              
      {\pnet{
      \pnformulae{~\pnf[?b]{\small $\wn\wn \bot$}~~~
      	\pnf[!A-oA]{\small$\!\!\oc(A^{\bot} \!\parr\! A) \quad $}~
      }
      \pndaimon[d]{?b,!A-oA}
      \pnsemibox{?b,!A-oA,d}
      }}\Biggr\}
    	&\pnrewrite[\Der_1]&
      \tikzsetnextfilename{images/long-path-poly-3}
      \Biggl\{
      \scalebox{\smallproofnets}                              
      {\pnet{
      \pnformulae{~\pnf[?b]{\small $\wn \bot$}~~~
      \pnf[!A-oA]{\small$\!\!\oc (A^{\bot} \!\parr\! A)\quad$}~}
      \pndaimon[d]{?b,!A-oA}
      \pnsemibox{?b,!A-oA,d}
      }}\Biggr\}\\
      R=&\tikzsetnextfilename{images/long-path-mell-1}
      \scalebox{\smallproofnets}{\pnet{
      \pnformulae{~~\pnf[a1]{\small $A^{\bot}$}~\pnf[a2]{\small $A$}~\\
      \pnf[b1]{\small $\bot$}~\pnf[b2]{\small $\bot$}}
      \pnaxiom[axa]{a1,a2}
      \pnpar{a1,a2}[A-oA]{\small $A^{\bot} \!\parr\! A$}
      \pnbag{}{A-oA}[!A-oA]{\small $\oc (A^{\bot} \parr A)$}
      \pnbot[bot1]{b1}
      \pnbot[bot2]{b2}
      \pnexp{}{b1,b2}[wbot]{\small $\wn \bot$}[1.5]]
      \pnbox{axa,b1,b2,A-oA,a1,a2}
      \pnbox{axa,b1,b2,!A-oA,wbot,a1,a2}
      \pnexp{}{wbot}[wwbot]{\small $\wn \wn \bot$}[1.2]
      \pnprom{!A-oA}[!!A-oA]{\small $\oc \oc (A^{\bot} \parr A)$}
      \pnsemibox{wbot,b1,b2,a1,a2,axa,!!A-oA}
      }}
    	&\pnrewrite[\BoxR_2]&
      \tikzsetnextfilename{images/long-path-mell-2}
      \scalebox{\smallproofnets}{\pnet{
      \pnformulae{~~\pnf[a1]{\small $A^{\bot}$}~\pnf[a2]{\small $A$}~\\
      \pnf[b1]{\small $\bot$}~\pnf[b2]{\small $\bot$}}
      \pnaxiom[axa]{a1,a2}
      \pnpar{a1,a2}[A-oA]{\small $A^{\bot} \!\parr\! A$}
      \pnbot[bot1]{b1}
      \pnbot[bot2]{b2}
      \pnexp{}{b1,b2}[wbot]{\small $\wn \bot$}[1.5]
      \pnbox{axa,b1,b2,A-oA,a2}
      \pnprom{A-oA}[!A-oA]{\small $\oc (A^{\bot} \!\parr\! A)$}
      \pnexp{}{wbot}[wwbot]{\small $\wn \wn \bot$}
      \pnsemibox{wbot,wwbot,b1,b2,a1,a2,axa}
      }}
    	&\pnrewrite[\Der_1]&
      \tikzsetnextfilename{images/long-path-mell-3}
      \scalebox{\smallproofnets}{\pnet{
      \pnformulae{~~\pnf[a1]{\small $A^{\bot}$}~\pnf[a2]{\small $A$}~\\
      	\pnf[b1]{\small $\bot$}~\pnf[b2]{\small $\bot$}}
      \pnaxiom[axa]{a1,a2}
      \pnpar{a1,a2}[A-oA]{\small $A^{\bot} \!\parr\! A$}
      \pnbot[bot1]{b1}
      \pnbot[bot2]{b2}
      \pnexp{}{b1,b2}[wbot]{\small $\wn \bot$}[1.5]
      \pnbox{axa,b1,b2,A-oA,a2}
      \pnprom{A-oA}[!A-oA]{\small $\oc (A^{\bot} \!\parr\! A)$}
            \pnsemibox{wbot,b1,b2,a1,a2,axa}
      }}
    	\\
    	\hline
      \pnrewrite[\BoxR_2]&
      \tikzsetnextfilename{images/long-path-poly-4}
      \Biggl\{\scalebox{\smallproofnets}{\pnet{
      \pnformulae{~\pnf[?b]{\small $\wn \bot$}~~~\pnf[A]{\small $\!A^\bot \!\parr\! A$}~}
      \pndaimon[d]{?b,A}
      \pnsemibox{?b,d,A}
      }}\Biggr\}
    	&\pnrewrite[\parr_2]&
      \tikzsetnextfilename{images/long-path-poly-5}
      \Biggl\{\scalebox{\smallproofnets}{\pnet{
      \pnformulae{~\pnf[?b]{\small $\wn \bot$}~~
      \pnf[Ab]{\small$A^{\bot}$}~\pnf[A]{\small $A$}~}
      \pndaimon[d]{?b,Ab,A}
      \pnsemibox{?b,d,A}
      }}\Biggr\}
      &\pnrewrite[\Mix_1]&                                
      \tikzsetnextfilename{images/long-path-poly-6}
      \Biggl\{\scalebox{\smallproofnets}{\pnet{
      \pnformulae{~\pnf[?b]{\small $\wn \bot$}~~~
      \pnf[Ab]{\small$A^{\bot}$}~\pnf[A]{\small $A$}~}
      \pndaimon[daib]{?b}
      \pndaimon[daiA]{Ab,A}
      \pnsemibox{daib,?b}
      \pnsemibox{daiA,Ab,A}                     
      }}\Biggr\}\\        
      \pnrewrite[\BoxR_2]&
      \tikzsetnextfilename{images/long-path-mell-4}
      \scalebox{\smallproofnets}{\pnet{
      \pnformulae{
      \pnf[b1]{\small $\bot$}~\pnf[b2]{\small $\bot$} ~\pnf[a1]{\small
                          $A^{\bot}$}~\pnf[a2]{\small $A$}~}
      \pnaxiom[axa]{a1,a2}
      \pnpar[A-oA]{a1,a2}{\small $A^\bot \!\parr\! A$}
      \pnbot[bot1]{b1}
      \pnbot[bot2]{b2}
      \pnexp{}{b1,b2}[wbot]{\small $\wn \bot$}
      \pnsemibox{bot1,wbot,a1,a2,axa,A-oA}
      }}
    	&\pnrewrite[\parr_2]&
      \tikzsetnextfilename{images/long-path-mell-5}
      \scalebox{\smallproofnets}{\pnet{
      \pnformulae{
      \pnf[b1]{\small $\bot$}~\pnf[b2]{\small $\bot$} ~\pnf[a1]{\small
                          $A^{\bot}$}~\pnf[a2]{\small $A$}~}
      \pnaxiom[axa]{a1,a2}
      \pnbot[bot1]{b1}
      \pnbot[bot2]{b2}
      \pnexp{}{b1,b2}[wbot]{\small $\wn \bot$}     
      \pnsemibox{bot1,wbot,a1,a2,axa}                   
      }}&\pnrewrite[\Mix_1]&
      \tikzsetnextfilename{images/long-path-mell-6}
      \scalebox{\smallproofnets}{\pnet{
      \pnformulae{
      \pnf[b1]{\small $\bot$}~\pnf[b2]{\small $\bot$}~~~\pnf[a1]{\small
                          $A^{\bot}$}~\pnf[a2]{\small $A$}~}
      \pnaxiom[axa]{a1,a2}
      \pnbot[bot1]{b1}
      \pnbot[bot2]{b2}
      \pnexp{}{b1,b2}[wbot]{\small $\wn \bot$}
      \pnsemibox{bot1,bot2,wbot}                       
      \pnsemibox{axa,a1,a2}
      }}
    \end{array}
  \end{align*}
   \vspace*{-\baselineskip}
  \caption{The path
    \texorpdfstring{$\BoxR_2 \, \Der_1 \, \BoxR_2 \, \parr_2 \, \Mix_1 \,
      \Ax_{2,3} \, \Contr_1 \, \Der_2 \, \Mix_1 \, \bot_2 \, \Der_1 \, \bot_1$}{} witnessing that
    \texorpdfstring{$\rho \in \Taylor{R}$}{} (to be continued on \Cref{fig:longex2}).}
  \label{fig:longex1}
\end{figure}

\begin{figure}[t]
  \centering
  \vspace{\beforepn}
  \begin{align*}
    \begin{array}{cccccccccc}
      \pnrewrite[\Ax_{2,3}]&
      \tikzsetnextfilename{images/long-path-poly-7}
      \Biggl\{\scalebox{\smallproofnets}{\pnet{
      \pnformulae{~\pnf[?b]{\small $\wn \bot$}~}
      \pndaimon[d]{?b}
      \pnsemibox{d,?b}
      }}\Biggr\}&\pnrewrite[\Contr_1]
      &\tikzsetnextfilename{images/long-path-poly-8}
      \Biggl\{\scalebox{\smallproofnets}{\pnet{
      \pnformulae{~\pnf[b1]{\small $\wn \bot$}~\pnf[b2]{\small $\wn \bot$}~}
      \pndaimon[d]{b1,b2}
            \pnsemibox{d,b1,b2}
      }}\Biggr\}&\pnrewrite[\Der_2]
      &\tikzsetnextfilename{images/long-path-poly-9}
      \Biggl\{\scalebox{\smallproofnets}{\pnet{
      \pnformulae{~\pnf[b1]{\small $\wn \bot$}~\pnf[b2]{\small $\bot$}~}
      \pndaimon[d]{b1,b2}
            \pnsemibox{d,b1,b2}
      }}\Biggr\}&\pnrewrite[\Mix_1]
      &\tikzsetnextfilename{images/long-path-poly-10}
      \Biggl\{\scalebox{\smallproofnets}{\pnet{
      \pnformulae{~\pnf[b1]{\small $\wn \bot$}~~~\pnf[b2]{\small $\bot$}~}
      \pndaimon[db1]{b1}
      \pndaimon[db2]{b2}
      \pnsemibox{db1,b1}  
      \pnsemibox{db2,b2}
      }}\Biggr\}\\
      \pnrewrite[\Ax_{2,3}]&
      \tikzsetnextfilename{images/long-path-mell-7}
      \scalebox{\smallproofnets}{\pnet{
      \pnformulae{
      \pnf[b1]{\small $\bot$}~\pnf[b2]{\small $\bot$}}
      \pnbot[bot1]{b1}
      \pnbot[bot2]{b2}
      \pnexp{}{b1,b2}[wbot]{\small $\wn \bot$}
      \pnsemibox{bot1,bot2,b1,b2,wbot}
      }}&\pnrewrite[\Contr_1]
      &\tikzsetnextfilename{images/long-path-mell-8}
      \scalebox{\smallproofnets}{\pnet{
      \pnformulae{
      \pnf[b1]{\small $\bot$}~\pnf[b2]{\small $\bot$}}
      \pnbot[bot1]{b1}
      \pnbot[bot2]{b2}
      \pnexp{}{b1}[wb1]{\small $\wn \bot$}
      \pnexp{}{b2}[wb2]{\small $\wn \bot$}
      \pnsemibox{bot1,bot2,b1,b2,wb1,wb2}
      }}&\pnrewrite[\Der_2]
      &\tikzsetnextfilename{images/long-path-mell-9}
      \scalebox{\smallproofnets}{\pnet{
      \pnformulae{
      \pnf[b1]{\small $\bot$}~\pnf[b2]{\small $\bot$}}
      \pnbot[bot1]{b1}
      \pnbot[bot2]{b2}
      \pnexp{}{b1}[wb1]{\small $\wn \bot$}
      \pnsemibox{bot1,bot2,b1,b2,wb1}
      }}&\pnrewrite[\Mix_1]&\tikzsetnextfilename{images/long-path-mell-10}
      \scalebox{\smallproofnets}{\pnet{
      \pnformulae{
      \pnf[b1]{\small $\bot$}~~~\pnf[b2]{\small $\bot$}}
      \pnbot[bot1]{b1}
      \pnbot[bot2]{b2}
      \pnexp{}{b1}[wb1]{\small $\wn \bot$}
      \pnsemibox{bot1,b1,wb1}
      \pnsemibox{bot2,b2}
      }}\\
    	\hline
      \pnrewrite[\bot_2]
      &\tikzsetnextfilename{images/long-path-poly-11}
      \Biggl\{\scalebox{\smallproofnets}{\pnet{
      \pnformulae{~\pnf[b1]{\small $\wn \bot$}~}
      \pndaimon[d]{b1}
      \pnsemibox{d,b1}
      }}\Biggr\}&\pnrewrite[\Der_1]
      &\tikzsetnextfilename{images/long-path-poly-12}
      \Biggl\{\scalebox{\smallproofnets}{\pnet{
      \pnformulae{~\pnf[b1]{\small $\bot$}~}
      \pndaimon[d]{b1}
      \pnsemibox{d,b1}
      }}\Biggr\}&\pnrewrite[\bot_1]
      &\left\{\quad\emptynet\quad\right\}\\
      \pnrewrite[\bot_2]
      &\tikzsetnextfilename{images/long-path-mell-11}
      \scalebox{\smallproofnets}{\pnet{
      \pnformulae{
      \pnf[b1]{\small $\bot$}}
      \pnbot[bot1]{b1}
      \pnexp{}{b1}[wb1]{\small $\wn \bot$}
      \pnsemibox{b1,bot1}
      }}&\pnrewrite[\Der_1]
      &\tikzsetnextfilename{images/long-path-mell-12}
      \scalebox{\smallproofnets}{\pnet{
      \pnformulae{
      \pnf[b1]{\small $\bot$}}
      \pnbot[bot1]{b1}
      \pnsemibox{bot1,b1}
      }}&\pnrewrite[\bot_1] & \left\{\quad\emptynet\quad\right\}
    \end{array}
  \end{align*}
  \vspace{-\baselineskip}
  \caption{The path
    \texorpdfstring{$\BoxR_2 \, \Der_1 \, \BoxR_2 \, \parr_2 \, \Mix_1 \,
      \Ax_{2,3} \, \Contr_1 \, \Der_2 \, \Mix_1 \, \bot_2 \, \Der_1 \, \bot_1$}{} witnessing that
    \texorpdfstring{$\rho \in \Taylor{R}$}{} (continued from \Cref{fig:longex1}).}
  \label{fig:longex2}
\end{figure}


%% file: general-case.tex
\section{Non-atomic axioms}
\label{sec:non-atomic}
From now on, we relax the definition of quasi-proof-structure (\Cref{def:proof-structure} and \Cref{fig:resource-cells}) so that the outputs of any $\ax$-cell are labeled by dual \MELL formul\ae, not necessarily atomic.
We can extend our results to this more general setting, with
some technical complications. 
Indeed, 
the rewrite rule for contraction has to be modified. Consider a set of \polyadic
proof-structures consisting of just a singleton which is a $\maltese$-cell. 
The contraction rule rewrites it as:
\vspace{\beforepn}
\begin{align*}
  \scalebox{\magicnumber}{
  \pnet{
  \pnformulae{
  \pnf[!A1]{$\oc A^{\bot}$}~\pnf[!A2]{$\oc A^{\bot}$}~\pnf[?A]{$\wn A$}
  }
  \pndaimon[d]{!A1,!A2,?A}[1][.2]
  }} \pnrewrite[\Contr_3]
  \tikzsetnextfilename{images/daimon-contraction-rewrite-2}
  \Big\{
  \scalebox{\magicnumber}{
  \pnet{
  \pnformulae{
  \pnf[!A1]{$\oc A^{\bot}$}~\pnf[!A2]{$\oc A^{\bot}$}~\pnf[?A1]{$\wn A$}~
  \pnf[?A2]{$\wn A$}
  }
  \pndaimon[d]{!A1,!A2,?A1,?A2}
  }}
  \Big\}
\textup{ which is then in the Taylor expansion of }
  \scalebox{\magicnumber}{
  \pnet{
  \pnformulae{
  \pnf[!A1]{$\oc A^{\bot}$}~\pnf[!A2]{$\oc A^{\bot}$}~\pnf[?A1]{$\wn A$}~
  \pnf[?A2]{$\wn A$}
  }
  \pnaxiom{!A1,?A2}[1.5]
  \pnaxiom{!A2,?A1}
  }}
\end{align*}
on which no contraction rewrite rule $\Contr$ can be applied backwards, breaking the naturality.
The failure of the naturality is actually due to the failure of
\Cref{prop:unwinding-co-functional} in the case of the rewrite rule \(\Contr\):
\(\op{\pnrewrite[\Contr]}\) (\ie $\pnrewrite[\Contr]$ read from the right to the left) is functional but not total.

The solution to this conundrum lies in changing the contraction rule for
\polyadicdaimon quasi-proof-structures, by explicitly adding \(\wn\)-cells. Hence, the
application of a contraction step $\Contr$ in the \polyadicdaimon quasi-proof-structures precludes the possibility of anything else but a \(\wn\)-cell on the \(\MELLdaimon\) side, which
allows the contraction step $\Contr$ to be applied backwards.

In turn, this forces us to change the definition of the \filled Taylor
expansion into a \emph{\(\eta\)-\filled Taylor expansion}, which has to include
elements where a $\maltese$-cell (representing an empty component) has some of its
outputs connected to \(\wn\)-cells.


\begin{definition}[\(\eta\)-\filled Taylor expansion]
  An \emph{$\eta$-emptying} of a \polyadic quasi-proof-structure
\(\rho = (|\rho|, \TreeT, \BoxFunction)\) is a \polyadic quasi-proof-structure with the same conclusions as in \(\rho\), obtained from $\rho$ by replacing each of the components of some roots of $\TreeT$ with a $\maltese$-cell whose outputs are either tails or inputs of a $\wn$-cell whose output $i$ is a tail, provided that $i$ is the output tail of a $\wn$-cell in $\rho$.

  The \emph{\(\eta\)-\filled Taylor expansion} \(\FatTaylorEta{R}\) of a
  quasi-proof-structure $R$ is the set of all the $\eta$-emptyings of every element of
  its Taylor expansion $\Taylor{R}$. 
\end{definition}

Note that the \(\eta\)-\filled Taylor expansion contains all the elements of
the \filled Taylor expansion and some more, such as the one in
\Cref{fig:eta-emptyings-text}.

\begin{figure}[!t]
  \vspace{\beforepn}
  \centering
  \scalebox{0.75}{
    \tikzsetnextfilename{images/eta-emptying-1}
    \pnet{
      \pnformulae{
        \pnf[Xp]{$X^\bot$}~~\pnf[tens]{$X\otimes \wn \bot$}
        ~~\pnf[!1]{$\oc 1$}~~~\pnf[yg]{}~\pnf[YY]{$\oc (\wn Y \parr Y^\bot)$}~~\pnf[1']{$1$}~\pnf[1'']{$1$}\\
      }
      \pndaimon[dai]{Xp,tens,!1}
      \pnexp{}{!1}[?!1]{$\wn\oc 1$}
      \pncown[!par]{YY}
      \pnone[1cell']{1'}
      \pnone[1cell'']{1''}
      \pnbag{}{1',1''}[!1']{$\oc 1$}
      \pnsemiboxc{dai,Xp,tens, !1} 
      {green}
      \pnsemibox{YY,1cell',1cell'',!par,yg} 
    }
  } \qquad
	\
	\begin{forest}
		for tree={grow'=north}
		[$\textcolor{green}{\bullet}$]
	\end{forest}
\!\!
	\begin{forest}
		for tree={grow'=north}
		[$\textcolor{gray}{\bullet}$]
	\end{forest}
  \caption{An element of the \(\eta\)-\filled Taylor expansion of the \mbox{\MELL
    		 quasi-proof-structure in Fig.\,\ref{fig:pointed-proof-net}}.}
  \label{fig:eta-emptyings-text}
\end{figure}

%
Functors $\qMELL$ and \(\PPolyPN\) are defined as before (Def.~\ref{def:FunctorqMELL} and \ref{def:FunctorPPoly}, respectively),\footnotemark
\footnotetext{Remember that now, for $\Gamma$ a list of list of $\MELL$ formul\ae{}, $\qMELL(\Gamma)$ (resp.~$\PolyPN(\Gamma)$) is the set of $\MELLdaimon$ (resp.~\polyadicdaimon) quasi-proof-structures of type $\Gamma$, possibly with non-atomic axioms.}
except that the image of \(\PPolyPN\) on the generator \(\Contr_i\) (\Cref{fig:daimon-tensor}) is changed to
\vspace{\beforepn}
\begin{align*}
\tikzsetnextfilename{images/daimon-contraction-rewrite-1}
\scalebox{\magicnumber}{
	\pnet{
		\pnformulae{
			\pnf[i]{$\MMaybe{\Gamma_k}$}~~\pnf[l]{$i$}
		}
		\pndaimon[d]{i,l}[1][.2]
		\pnsemibox{d,i}
}} \pnrewrite[\Contr_i]
\tikzsetnextfilename{images/daimon-contraction-rewrite-2}
\Bigg\{
\scalebox{\magicnumber}{
	\pnet{
		\pnformulae{
			\pnf[i]{$\MMaybe{\Gamma_{k}}$}~\pnf[n]{\quad}~\pnf[n']{\quad}~
		}
		\pndaimon[d]{i,n,n'}
		\pnexp{}{n}[?n]{$i$}
		\pnexp{}{n'}[?n']{$i\!+\!1$}
		\pnsemibox{d,i,n,n'}
}}
\Bigg\}
\end{align*}

\vspace{-.5\baselineskip}
\noindent where \(\MMaybe{\Gamma_k}\) signifies that some of the conclusions of \(\Gamma_k\)
might be connected to the \(\maltese\)-cell through a \(\wn\)-cell\SLV{}{ (see \Cref{app:general} for details)}.
We can prove similarly our main results.

\begin{theorem}[naturality with $\eta$]
  The \(\eta\)-\filled Taylor expansion defines a natural transformation
  $\NatTransf_{\eta} \colon \PPolyPN \To \qMELL \colon \Path \!\to\! \Rel$
  by: 
  $(\Pi, R) \!\in\! \NatTransfEta_{\Gamma}$ iff $\Pi \!\subseteq\! \FatTaylorEta{R}$ and the type of $R$ is $\Gamma$.
	Moreover, if $\Pi$ is a set of \polyadic proof-structures with $\Pi \pnrewrite[\xi] \Pi'$ and $\Pi' \subseteq \Taylor{R'}$, then $R$ is a \MELL proof-structure and $\Pi \subseteq \Taylor{R}$, where $R$ is such that $R \pnrewrite[\xi] R'$.
\end{theorem}
\begin{theorem}[glueability criterion with $\eta$]
  Let \(\Pi\) be a set of \polyadic proof-structures, not necessarily with
  atomic axioms:
  \(\Pi\) is glueable iff $\Pi \subseteq \Taylor{R}$ for some  \(\MELL\) proof-structure~$R$.
%
\end{theorem}

%% file: conclusion.tex
\section{Conclusions and perspectives}
\label{sec:conclusions}

\subparagraph{\texorpdfstring{\maltese}{daimon}-cells inside boxes}

Our glueability criterion (\Cref{thm:characterization}) solves the inverse Taylor expansion problem in a ``asymmetric'' way: we characterize the sets of \polyadic proof-structures that are included in the Taylor expansion of some \MELL proof-structure, but \polyadic proof-structures have no occurrences of $\maltese$-cells, while a \MELL proof-structure possibly contains $\maltese$-cells inside boxes (see \Cref{def:proof-structure}).
Not only this asymmetry is technically inevitable, but it reflects on the fact that some glueable set of \polyadic proof-structure might not contain any information on the content of some box (which is reified in \MELL by a \(\maltese\)-cell), or worse that, given the types, no content can fill that box.
Think of the \polyadic proof-structure $\rho$ made only of a $\oc$-cell with no inputs and one output of type $\oc X$, where $X$ is atomic: $\{\rho\}$ is glueable but 
the only \MELL proof-structure $R$ such that $\{\rho\} \subseteq \Taylor{R}$ is made of a box containing a $\maltese$-cell.

This asymmetry is also present in Pagani and Tasson's characterization \cite{PaganiTasson2009}, even if not particularly 
emphasized: their Theorem 2 (analogous to the left-to-right part of our \Cref{thm:characterization}) assumes not only that
the rewriting starting from a finite set of \polyadic proof-structures terminates 
but also that it ends on a \MELL proof-structure (without $\maltese$-cells, which ensures that there exists a \MELL proof-structure without $\maltese$-cells filling all the empty boxes).

\subparagraph{The \texorpdfstring{\(\lambda\)}{lambda}-calculus, connectedness and
  coherence}
Our rewriting system and glueability criterion might help to
prove that a binary coherence relation can solve the inverse Taylor expansion problem for \(\MELL\) proof-structures fulfilling some geometric property related to connectedness, despite the impossibility for the full \MELL fragment. 
Such a coherence would extend the coherence criterion for resource \(\lambda\)-terms. 
Note that our glueability criterion is
actually an extension of the criterion for resource \(\lambda\)-terms.
Indeed,
in the case of the \(\lambda\)-calculus, there are three rewrite steps,
corresponding to 
abstraction, application and variable (which can be
encoded in our rewrite steps), and coherence is defined inductively: if a set
of resource \(\lambda\)-terms is coherent, then any set of resource
\(\lambda\)-term that rewrites to it is also coherent.

Presented in this way, the main difference between the $\lambda$-calculus and ``connected'' $\MELL$ (concerning the inverse Taylor expansion problem) would not be because of the rewriting system but because the structure of any resource \(\lambda\)-term univocally determines the rewriting path, 
while, for \polyadic
proof-structures, we have to quantify existentially over all possible paths.
This is an unavoidable consequence of the fact that proof-structures
do not have a tree-structure, contrary to \(\lambda\)-terms and resource $\lambda$-terms.

Moreover, it is possible to match and mix different sequences of rewriting.
Indeed, consider three \polyadic proof-structures pairwise glueable. Proving
that they are glueable as a whole amounts to computing a rewriting path from the
rewriting paths witnessing the three glueabilities. Our paths were designed with
that mixing-and-matching operation in mind, in the particular case where the
boxes are connected. This is reminiscent of \cite{FSCD2016}, where we also
showed that a certain property enjoyed by the \(\lambda\)-calculus can be
extended to proof-structures, provided they are connected inside boxes. We leave
that work to a subsequent paper.

\subparagraph{Functoriality and naturality}

Our functorial point of view on proof-structures might unify many results. Let us
cite two of them:
\begin{itemize}
\item a sequent calculus proof of \(\vdash \Gamma\) can be translated into a
  path from the empty sequence into \(\Gamma\). 
  This could be the starting point for the formulation of a new correctness criterion;
\item the category \(\Path\) can be extended with higher structure, allowing to
  represent cut-elimination. The functors \(\qMELL\) and \(\PPolyPN\) can also
  be extended to such higher functors, proving 
	via naturality that cut-elimination and the Taylor expansion commute.
\end{itemize}



\subparagraph{Acknowledgments} This work has been partially funded by the EPSRC
grant EP/R029121/1 ``Typed Lambda-Calculi with Sharing and Unsharing'' and the ANR
project Rapido (ANR-14-CE35-0007).
 

%% file: naturality-appendix.tex
\section{Proof of naturality (\Cref{thm:projection-natural}, p.~\pageref{thm:projection-natural})}
\label{sec:naturality-proof}

  $\Projection$ is a family of morphisms of $\Rel$ indexed by the objects in
  $\Path$. In the first part of the statement (before ``Moreover''), the only thing to show is that the naturality squares
  \tikzsetnextfilename{images/naturality-projection}
  \begin{commutativediagram}{5}{5}{5}
    {
      \PPolyPN(\Gamma) \& \PPolyPN(\Gamma')\\
      \qMELL(\Gamma) \& \qMELL(\Gamma')\\
    };
    \path[-stealth]
    (m-1-1) edge [auto] node {$\PPolyPN(\xi)$} (m-1-2)
    (m-1-1) edge [auto] node {$\Projection_{\Gamma}$} (m-2-1)
    (m-2-1) edge [auto] node {$\qMELL(\xi)$} (m-2-2)
    (m-1-2) edge [auto] node {$\Projection_{\Gamma'}$} (m-2-2);
  \end{commutativediagram}
  \noindent commute for every path $\xi:\Gamma\to \Gamma'$, and actually, it is enough to
  show that such squares commute for all elementary paths.
  Let $a \colon \Gamma\to \Gamma'$ be an elementary path.

  \begin{enumerate}
  \item
    $\qMELL(a) \circ \Projection_{\Gamma} \subseteq \Projection_{\Gamma'} \circ
    \PPolyPN(a)$.

    Let $(\Pi,R') \in \qMELL(a) \circ \Projection_{\Gamma}$. Let $R\in
    \qMELL(\Gamma)$ be a witness of composition, that is an element such that
    $(\Pi,R) \in \Projection_{\Gamma}$ and $R\pnrewrite[a] R'$.

    Let $p \colon t\to \TreeT$ be a thick subforest of $\TreeT$, let $r_{1},\ldots,r_{n}$ be some roots of $\TreeT$, and let $\rho_{r_{1}\ldots r_{n}} \in \Pi$ be the element of the
    filled Taylor expansion of $R$ associated with $p$ and $r_{1},\ldots,r_{n}$.
    \begin{itemize}
    \item If $a = \Mix_{i}$, then, in $R$, the conclusions $1,\ldots,i,i+1,\dots k$ are exactly the conclusions of a root in the
      box-forest of $R$, and the connected components in $R$ of $i$ and $i+1$ are disjoint. By Definitions~\ref{def:taylor} and~\ref{def:FilledTaylor}, since $\rho_{r_{1}\ldots r_{n}} \in \Pi\subseteq\FatTaylor{R}$, we have that the conclusions $1,\ldots,i,i+1,\dots k$ are exactly the conclusions of a root $r$ in the box-forest of $\rho_{r_{1}\ldots r_{n}}$, and we have two possibilities:
      \begin{itemize}
      \item
      the connected components of $i$ and $i+1$ are disjoint in $\rho_{r_{1}\ldots r_{n}}$;
      \item
      $i$ and $i+1$ belong to the same connected component, in which case $r\in\{r_{1},\ldots,r_{n}\}$ and $\rho_{r_{1}\ldots r_{n}}$ is a $\maltese$-cell with conclusion $1,\ldots,i,i+1,\dots k$.      
      \end{itemize}
      In both cases  the rule $\Mix_{i}$ is also applicable in $\rho_{r_{1}\ldots r_{n}}$, yielding a \polyadic proof-structure $\rho'$. The box-forest $\TreeT'$ of $R'$ is obtained from the box-forest $\TreeT$ of $R$ by replacing a root $b$ by two roots $b_1$, $b_2$. Let
      $p':t'\to \TreeT'$ be such that all the boxes $d \neq b_1, b_2$ have the
      same inverse image than by $p$: $p'^{-1}(d) = p^{-1}(d)$, and,
      $p'^{-1}(b_1) = p^{-1}(b) \times \{1\}$, $p^{-1}(b_2) = p^{-1}(b) \times
      \{2\}$. We verify that $\rho'$ is the filled Taylor expansion of $R'$ through
      $p'$.
%
%
%
%
    \item If $a\in \{\Ax_{i}, \maltese_{i}, \One_i, \bot_i, \Weak_i\}$, let
      $k$ be such that the rule $a$ acts on the conclusions $i,\dots,i+k$ in
      $R$, and let $\ell$ be the type of the cell in $R$ connected to the
      conclusions $i,\dots,i+k$. In $\rho_{r_{1}\ldots r_{n}}$ there is a cell of type $\ell$ or
      $\maltese$ connected to the same conclusions. Clearly $a$ is applicable to $\rho_{r_{1}\ldots r_{n}}$, which yields a
      \polyadic proof-structure $\rho'$.

      The box-forest $\TreeT'$ of $R'$ is obtained from the box-forest $\TreeT$ of $R$ by erasing a root $b$. Let $p':t'\to \TreeT'$ be such that all
      the boxes $d \neq b$ have the same inverse image than by $p$:
      $p'^{-1}(d) = p^{-1}(d)$. We verify that $\rho'$ is the filled Taylor expansion of
      $R'$ through $p'$.
    \item If $a \in \{\otimes_i, \parr_i, \Der_i,\Contr_i\}$, let $k$ be such that the
      rule $a$ acts on the conclusions $i,\dots,i+k$ in $R$, and let $\ell$ be the
      type of the cell in $R$ connected to the conclusions $i,\dots,i+k$. In
      $\rho_{r_{1}\ldots r_{n}}$ there is a cell of type $\ell$ or $\maltese$ connected to the same
      conclusions. Clearly $a$ is applicable to $\rho_{r_{1}\ldots r_{n}}$, which yields a
      \polyadic proof-structure $\rho'$.

      $R'$ has the same box-forest $\TreeT$ as $R$. We verify that $\rho'$ is
      the expansion of $R'$ through $p$.
    \item If $a = \Cut^i$, let $c$ be the $\Cut$-cell in $R$ to which the rule
      is applied. The $\Cut$-cell $c$ has either one image in $\rho_{r_{1}\ldots r_{n}}$ or is represented by a
      $\maltese$ cell. In both cases, $\Cut^i$ is applicable to $\rho_{r_{1}\ldots r_{n}}$, yielding
      $\rho'$.

      $R'$ has the same box-forest $\TreeT$ as $R$. We verify that $\rho'$ is
      the expansion of $R'$ through $p$.
%
%
    \item If $a = \BoxR_{i}$, let $k$ be such that the rule $a$ acts on the
      conclusions $i,\dots,i+k$ in $R$. In $\rho_{r_{1}\ldots r_{n}}$ we have one of the following possibilities:
      \begin{itemize}
      \item
      $\rho_{r_{1}\ldots r_{n}}$ consists of a unique $\maltese$-cell with the same conclusions $i,\dots,i+k$;
      \item
      $\rho_{r_{1}\ldots r_{n}}$ consists of a $\oc$-cell in $i$ with no premises and $k$ $\wn$-cells with no premises above the other $k$ conclusions;
      \item
      there is a $\oc$-cell above the conclusion $i$ and a $\wn$-cell above each of the other $k$ conclusions; and the other cells of this root can
      be identified by their image $1,\dots,\ell$ in $t$: we have $\ell$ pairwise 
      disconnected sub-proof-structures $\pi_1,\dots, \pi_{\ell}$. 
           \end{itemize}
       In any case, the rule $\BoxR_i$ can be applied, yielding either a family
      $\rho'_1,\dots,\rho'_{\ell}$ of \polyadicdaimon proof-structures or a \polyadicdaimon proof-structure
      $\rho'_1$. More precisely, in the first (resp.\ second, third) case, we apply the Daimonded (resp.\ Empty, Non-empty) box rule (see 
      Figure~\ref{fig:actions-daimon}(g),~\ref{fig:actions-daimon}(h),~\ref{fig:actions-daimon}(i)).
      
      The box-forest $\TreeT'$ of $R'$ is obtained from the box-forest $\TreeT$ of $R$ by erasing the root of the conclusions $i,\dots,i+k$: the new root $b'$ of this tree of $\TreeT'$ is the unique vertex 
      connected to the root of $\TreeT$ (its unique son). We have $p^{-1}(b')=\{b'_{1},\ldots,b'_{\ell}\}$, and $\ell$ trees $t'_{1},\ldots,t'_{\ell}$, where $b'_{i}$ is the root of $t'_{i}$. The morphisms 
      $p'_{i}:t'_{i}\to \TreeT'$ are defined accordingly, and $\rho'_i$ is the filled Taylor expansion of $R$ through $p'_{i}$.
    \end{itemize}
  \item
    $\Projection_{\Gamma'} \circ \PPolyPN(a) \subseteq \qMELL(a) \circ
    \Projection_{\Gamma}$.
    
    Let $(\Pi,R') \in \Projection_{\Gamma'} \circ \PPolyPN(a)$. Let $\Pi'$ be
    a witness of composition, that is a set of $\PPolyPN(\Gamma')$ such that
    $(\Pi,\Pi') \in \PPolyPN(a)$ and $(\Pi',R') \in \Projection_{\Gamma'}$.

    We want to exhibit a $\MELLdaimon$ quasi-proof-structure $R$ such that
    $R \pnrewrite[a] R'$ and $\Pi$ is a part of the \filled Taylor expansion of
    $R$. By co-functionality of $\qMELL(a)$
    (\Cref{prop:unwinding-co-functional}), we have a candidate for such an $R$:
    the pre-image of $R'$ by this co-functional relation. We only have to check that $R'$
    is in the image of the co-functional relation and that $\Pi$ is a part of the \filled
    Taylor expansion of the pre-image $R$ of $R'$ by the co-functional relation. In other terms: if $a \colon \Gamma \to \Gamma'$
    and $R'\in\qMELL(\Gamma')$, then there exists $R\in\qMELL(\Gamma)$ such that 
    $R\pnrewrite[a] R'$ and $\Pi\subseteq \FatTaylor{R}$.    
    \begin{itemize}
    \item If $a\neq \Contr_i$, there exists (a unique) $R\in\qMELL(\Gamma)$ such that 
    $R\pnrewrite[a] R'$. The case $a=\Contr_i$ is a bit more delicate: in this case too there exists (a unique)
    $R\in\qMELL(\Gamma)$ such that $R\pnrewrite[a] R'$, but here we use the fact that the types of the axiom conclusions are atomic. Indeed,
    thanks to this choice every conclusion of $R'$ of type $\wn A$ is the conclusion of a $\wn$-cell (or of a $\maltese$-cell).
        \item Let $R$ be the unique pre-image of $R'$ through $\qMELL(a)$
      ($R = \op{\pnrewrite[a]} (R')$). We need to show that $\Pi$ is a part of
      the \filled Taylor expansion of $R$. Let $\rho \in \Pi$ and
      $\{\rho'_1, \dots, \rho'_n\} \subseteq \Pi'$ such that
      $\rho \pnrewrite[a] \{\rho'_1, \dots, \rho'_n\}$. In all cases except
      $\BoxR$, this set is a singleton $\{\rho'_1\}$. In that cases, let
      $p' : t' \to \TreeT'$ and $r'_1, \dots, r'_k$ be the conclusions of $R'$ such that $\rho'_1=\rho_{r'_{1}\ldots r'_{k}}$ is the
      element of the filled Taylor expansion of $R'$ associated with $p'$ and $r'_{1},\ldots,r'_{k}$.

      If $a \in \{\Ax_{i}, \maltese_{i}, \One_i, \bot_i, \Weak_i \}$, let $r$
      be the root of the conclusion $i$ in $R$ and let $s$ be the root of the
      conclusion $i$ in $\rho$. $\TreeT$ is the disjoint union of $\TreeT'$ and
      the root $r$. Let $t$ be the disjoint union of $t'$ and the root $s$, and
      $p:t \to \TreeT$ be defined as $p'$ over $t'$ and $p(s)=r$. If the cell rooted in
      $i$ in $\rho$ is a $\maltese$, then we check that $\rho$ is the
      element of the filled Taylor expansion of $R$ associated with $p$ and $r,r'_{1},\ldots,r'_{k}$,
      else associated with $p$ and $r'_{1},\ldots,r'_{k}$.

      If $a \in \{\Cut^i, \otimes_i, \parr_i, \Der_i, \Contr_i\}$, then
      $\TreeT = \TreeT'$, $p=p'$ and we check that $\rho$ is the 
       element of the filled Taylor expansion of $R$ associated with $p$ and $r'_{1},\ldots,r'_{k}$.

      If $a = \Mix_{i}$, let $r'_1$ and $r'_2$ be the respective roots of
      the conclusion $i$ and $i+k$ in $R'$, and let $r$ be the root of the conclusion $i$
      in $R$. Consider the roots $s'_1$ and $s'_2$ in $t'$ such that $i$ is
      produced from $s'_1$ and $i+k$ from $s'_2$, let $t$ be equal to $t'$
      except that the two roots $s'_1$ and $s'_2$ are merged and change $p'$ into $p$ accordingly.
      We check that $\rho$ is the
       element of the filled Taylor expansion of $R$ associated with $p$ and $r'_{1},\ldots,r'_{k}$.

      If $a = \BoxR_{i}$, we describe the case of the Non-empty box rule (Figure~\ref{fig:actions-daimon}(i)), leaving the two easier cases 
      of the Daimonded (resp.\ Empty) box rule of Figure~\ref{fig:actions-daimon}(g) (resp.\ Figure~\ref{fig:actions-daimon}(h)) to the reader.\\ 
       Let $p'_{1}:t'_1\to\TreeT',\dots,p'_{n}:t'_n\to\TreeT'$ be such that
      $\rho'_1, \dots, \rho'_n$ are the expansions of $R'$ associated with, respectively,
      $p'_{1},\dots,p'_{n}$. The forests $t'_{1},\dots,t'_{n}$ differ by only one tree. Consider the
      forest $t$ which has all the trees on which the $n$ forests do not differ and
      the union of the trees on which the forests differ, all connected with a root, and define $p$ accordingly. We check that $\rho$ is the
       element of the filled Taylor expansion of $R$ associated with $p$ and $r'_{1},\ldots,r'_{k}$.
  \end{itemize}               
\end{enumerate}

Concerning the second part of the statement of \Cref{thm:projection-natural} (after ``Moreover''), 
we prove the following stronger statement: given two sets $\Pi$ and $\Pi'$ of \polyadicdaimon quasi-proof-structures and a \MELLdaimon quasi-proof-structure $R'$,
\begin{enumerate}
	\item if $\Pi \pnrewrite[a] \Pi'$ and $\Pi' \subseteq \Taylor{R'}$, then $\Pi \subseteq \Taylor{R}$ where $R$ is such that $R \pnrewrite[a] R'$;
	\item if moreover $\Pi$ is a set of \polyadic proof-structures, then $R$ is a \MELL proof-structure.
\end{enumerate}
Both points are proven by straightforward inspection of the rewrite rules defined in Figures \ref{fig:actions-all} and \ref{fig:actions-daimon}. 
The idea is that none of them, read from right to left, introduces a new $\maltese$-cell, thus from $\Pi' \subseteq \Taylor{R'}$ it follows that $\Pi \subseteq \Taylor{R}$;
whereas the only rewrite rule, read from left to right, that introduces a new $\maltese$-cell is the ``empty box'' one (\Cref{fig:empty-box}), so if $\Pi \pnrewrite[\BoxR_{i}] \Pi'$ according to that and $\Pi$ is a set of \polyadic proof-structures (in particular, no $\maltese$-cell occurs in any element of $\Pi$), then in $R$ the only occurrence of a $\maltese$-cell is necessarily the whole content of a box, hence $R$ is a \MELL quasi-proof-structure. 
Finally, $R$ is a \MELL proof-structure (without ``quasi'') because $\Pi \subseteq \Taylor{R}$ and the Taylor expansion preserves conclusions (\Cref{rmk:conclusions}).
\qed

%% file: not-eta-expansed.tex
\section{The general case: non-atomic axioms (\Cref{sec:non-atomic})}
\label{app:general}

We present here (\Cref{fig:eta-actions-daimon}) the rewrite rules for the general case of \polyadicdaimon quasi-proof-structures with possibly \emph{non-atomic axioms}. 
They replace the rewrite rules in \Cref{fig:actions-daimon} (which are valid only for \polyadicdaimon quasi-proof-structures with atomic axioms).
The rewrite rules in \Cref{fig:eta-actions-daimon} are essentially the same
as in \Cref{fig:actions-daimon},\footnotemark
\footnotetext{Up to the fact that here an output of a $\maltese$-cell is either a conclusion of the \polyadicdaimon quasi-proof-structure $\rho$ or the input of a $\wn$-cell whose output is a conclusion of $\rho$, see below.} 
except for the rewrite rule \(\Contr_i\) (\Cref{fig:daimon-tensor}), which is changed to \Cref{fig:eta-daimon-contraction}.

When representing a \polyadicdaimon quasi-proof-structure $\rho$, we write
\parbox[b][\baselineskip]{2cm}{\(
\tikzsetnextfilename{images/dai-wn}
\scalebox{\magicnumber}{
  \pnet{
    \pnformulae{
      \pnf[i]{$\MMaybe{i_1}$}~\pnf[l]{$\ldots$}~\pnf[j]{$\MMaybe{i_k}$}~\pnf[it]{$i$}
    }
    \pndaimon[dai]{i,j,it}
  }}
\)}
for a $\maltese$-cell whose outputs $i_1, \dots, i_k$ are either conclusions (as
$i$) of $\rho$, or inputs of $\wn$-cells whose outputs are conclusions of
$\rho$.

\begin{figure}[t]
  \centering
  \begin{subfigure}[c]{0.3\textwidth}
    \begin{align*}
      \tikzsetnextfilename{images/mix-daimon-rewrite-1}
      \scalebox{\magicnumber}{
      \pnet{
      \pnformulae{
      \pnf[i]{$\MMaybe{\Gamma_{k}}$}
      ~\pnf[l]{$\MMaybe{i}$}~\pnf[l1]{$\MMaybe{i\!+\!1}$}~~\pnf[j]{$\ \MMaybe{\Gamma_{k}'}  $}
      }
      \pndaimon[d]{i,l,l1,j}
      \pnsemibox{d,i,l,l1,j}
      }} \pnrewrite[\Mix_{i}]
      \tikzsetnextfilename{images/mix-daimon-rewrite-2}
      \Bigg\{
      \scalebox{\magicnumber}{
      \pnet{
      \pnformulae{
      \pnf[i]{$\MMaybe{\Gamma_{k}}$}~~\pnf[ik]{$\MMaybe{i}\ $}
      ~~\pnf[j1]{$\MMaybe{i\!+\!1}$}~
      ~\pnf[j]{$\MMaybe{\Gamma_{k}'}$}
      }
      \pndaimon[d1]{i,ik}
      \pndaimon[d2]{j1,j}
      \pnsemibox{d1,i,ik}
      \pnsemibox{d2,j1,j}
      }}
    \Bigg\}
    \end{align*}
    \vspace{-\baselineskip}
    \caption{Mix}
    \label{fig:eta-mix-daimon}
  \end{subfigure}%
  \begin{subfigure}[c]{0.3\textwidth}
    \begin{align*}
      \tikzsetnextfilename{images/daimon-ax-rewrite-1}
      \scalebox{\magicnumber}{
      \pnet{
      \pnformulae{
      \pnf{$\dots\,$}~\pnf[i]{$i$}~\pnf[j]{$i\!+\!1$}~\pnf{$\,\dots$}
      }
      \pndaimon[ax]{i,j}
      \pnsemibox{ax,i,j}
      }} \pnrewrite[\Ax_{i}]
      \tikzsetnextfilename{images/daimon-ax-rewrite-2}
      \Bigg\{
      \scalebox{\magicnumber}{
      \pnet{
      \pnformulae{
      \pnf{$\dots$}~\pnf{$\dots$}
      }
      }}
    \Bigg\}
    \end{align*}
    \vspace{-\baselineskip}
    \caption{Hypothesis (\(\Ax, \maltese, \One, \bot, \Weak\))}
    \label{fig:eta-ax-daimon}
  \end{subfigure}

  \begin{subfigure}[c]{0.3\textwidth}
    \begin{align*}
      \tikzsetnextfilename{images/daimon-cut-rewrite-1}
      \scalebox{\magicnumber}{
      \pnet{
      \pnformulae{
      \pnf[i]{$\MMaybe{\Gamma_k}$}~\pnf[l]{ }
      }
      \pndaimon[d]{i}[1][.5]
      \pnsemibox{d,i}
      }} \pnrewrite[\Cut^i]
      \tikzsetnextfilename{images/daimon-cut-rewrite-2}
      \Bigg\{
      \scalebox{\magicnumber}{
      \pnet{
      \pnformulae{
      \pnf[i]{$\MMaybe{\Gamma_{k}}$}~\pnf[n]{$i$}~\pnf[n']{$i\!+\!1$}~
      }
      \pndaimon[d]{i,n,n'}
      \pnsemibox{d,i,n,n'}
      }}
    \Bigg\}
    \end{align*}
    \vspace{-\baselineskip}
    \caption{Cut}
    \label{fig:eta-cut-daimon}
  \end{subfigure}%
  \begin{subfigure}[c]{0.3\textwidth}
    \begin{align*}
      \tikzsetnextfilename{images/daimon-tensor-rewrite-1}
      \scalebox{\magicnumber}{
      \pnet{
      \pnformulae{
      \pnf[i]{$\MMaybe{\Gamma_k}$}~~\pnf[l]{$i$}
      }
      \pndaimon[d]{i,l}[1][.2]
      \pnsemibox{d,i}
      }} \pnrewrite[\otimes_i]
      \tikzsetnextfilename{images/daimon-tensor-rewrite-2}
      \Bigg\{
      \scalebox{\magicnumber}{
      \pnet{
      \pnformulae{
      \pnf[i]{$\MMaybe{\Gamma_{k}}$}~\pnf[n]{$i$}~\pnf[n']{$i\!+\!1$}~
      }
      \pndaimon[d]{i,n,n'}
      \pnsemibox{d,i,n,n'}
      }}
    	\Bigg\}
    \end{align*}
    \vspace{-\baselineskip}
    \caption{Binary rule (\(\otimes, \parr\))}
    \label{fig:eta-daimon-tensor}
  \end{subfigure}%
  \begin{subfigure}[c]{0.2\textwidth}
    \begin{align*}
      \tikzsetnextfilename{images/daimon-dereliction-rewrite-1}
      \scalebox{\magicnumber}{
      \pnet{
      \pnformulae{
      \pnf[i]{\small $\MMaybe{\Gamma_k}$}~~\pnf[l]{$i$}
      }
      \pndaimon[d]{i,l}
      \pnsemibox{d,i,l}
      }} \pnrewrite[\Der_i]
      \tikzsetnextfilename{images/daimon-dereliction-rewrite-2}
      \Bigg\{
      \scalebox{\magicnumber}{
      \pnet{
      \pnformulae{
      \pnf[i]{\small $\MMaybe{\Gamma_{k}}$}~~\pnf[n']{\small $i$}~
      }
      \pndaimon[d]{i,n'}
      \pnsemibox{d,i,n'}
      }}
    \Bigg\}
    \end{align*}
    \vspace{-\baselineskip}
    \caption{Dereliction}
    \label{fig:eta-daimon-dereliction}
  \end{subfigure}

  \begin{subfigure}[b]{0.2\textwidth}
    \tikzsetnextfilename{images/daimon-box-rewrite-1}
    \begin{align*}
      \scalebox{\magicnumber}{
      \pnet{
      \pnformulae{
      \pnf[i1]{$\MMaybe{\Gamma}$}~\pnf[i]{$i$}
      }
      \pndaimon[dai]{i1,i}
      \pnsemibox{dai,i1,i}}
      } \pnrewrite[\BoxR_i]
      \tikzsetnextfilename{images/daimon-box-rewrite-1}
      \Bigg\{
      \scalebox{\magicnumber}{
      \pnet{
      \pnformulae{
      \pnf[i1]{$\MMaybe{\Gamma}$}~\pnf[i]{$i$}
      }
      \pndaimon[dai]{i1,i}
      \pnsemibox{dai,i1,i}}
      }
    	\Bigg\}
    \end{align*}
    \vspace{-\baselineskip}
    \caption{Daimoned box}
    \label{fig:eta-daimon-box}
  \end{subfigure}%
  \begin{subfigure}[b]{0.2\textwidth}
    \begin{align*}
      \tikzsetnextfilename{images/daimon-contraction-rewrite-1}
      \scalebox{\magicnumber}{
      \pnet{
      \pnformulae{
      \pnf[i]{$\MMaybe{\Gamma_k}$}~~\pnf[l]{$i$}
      }
      \pndaimon[d]{i,l}[1][.2]
      \pnsemibox{d,i}
      }} \pnrewrite[\Contr_i]
      \tikzsetnextfilename{images/daimon-contraction-rewrite-2}
      \Bigg\{
      \scalebox{\magicnumber}{
      \pnet{
      \pnformulae{
      \pnf[i]{$\MMaybe{\Gamma_{k}}$}~\pnf[n]{\quad}~\pnf[n']{\quad}~
      }
      \pndaimon[d]{i,n,n'}
      \pnexp{}{n}[?n]{$i$}
      \pnexp{}{n'}[?n']{$i\!+\!1$}
      \pnsemibox{d,i,n,n',?n,?n'}
      }}
      \Bigg\}
    \end{align*}
    \vspace{-\baselineskip}
    \caption{Contraction}
    \label{fig:eta-daimon-contraction}
  \end{subfigure}%
  \begin{subfigure}[b]{0.2\textwidth}
    \tikzsetnextfilename{images/empty-box-rewrite-1}
    \begin{align*}
      \scalebox{\magicnumber}{
      \pnet{
      \pnformulae{
      \pnf[Gamma2]{$\wn\Gamma$}~
      \pnf[i]{$i$}
      }
      \pnwn[?2]{Gamma2}
      \pncown[!i]{i}
      \pnsemibox{Gamma2,?2,!i}
      }}
      \pnrewrite[\BoxR_i]
      \tikzsetnextfilename{images/empty-box-rewrite-2}
      \Bigg\{
      \scalebox{\magicnumber}{
      \pnet{
      \pnformulae{
      \pnf[Gamma2]{$\wn\Gamma$}~
      \pnf[i]{$i$}
      }
      \pndaimon[dai]{i,Gamma2}
      \pnsemibox{dai,i,Gamma2}
      }
      }
    	\Bigg\}
    \end{align*}
    \vspace{-\baselineskip}
    \caption{Empty box}
    \label{fig:eta-empty-box}   
  \end{subfigure}

  \begin{subfigure}[b]{0.3\textwidth}
    \begin{align*}
      \tikzsetnextfilename{images/oc-split-1}
      \scalebox{\magicnumber}{
      \pnet{
      \pnformulae{
      \pnf{$\dots$}}
      \pnsomenet[pin]{$\rho_n$}{1.5cm}{0.5cm}[at (0,1.5)]
      \pnsomenet[pi1]{$\rho_1$}{1.5cm}{0.5cm}[at (0,0)]
      \pnoutfrom{pi1.-45}[pi1c]{$\quad$}
      \pnoutfrom{pin.-30}[pinc]{$\quad$}
      \pnbag[!]{}{pi1c,pinc}{$i$}[1][0.5]
      \pnoutfrom{pi1.-135}[pi1a1]{$\quad$}
      \pnoutfrom{pin.-150}[pina2]{$\quad$}
      \pnexp[?2]{}{pi1a1,pina2}{$\wn\Gamma_{k}$}[1][-0.5]
      \pnsemibox{pi1,pin,!,?2}
      }}
      \pnrewrite[\BoxR_i]
      \tikzsetnextfilename{images/oc-split-2}
      \left\{
      \scalebox{\magicnumber}{
      \pnet{
      \pnformulae{
      \pnf{$\dots$}}
      \pnsomenet[pi1]{$\rho_j$}{1.5cm}{0.5cm}[at (0,0)]
      \pnoutfrom{pi1.-30}[pi1c]{$i$}[2]
      \pnoutfrom{pi1.-150}[pi1a1]{$\quad$}
      \pnexp[?2]{}{pi1a1}{$\wn \Gamma_{k}$}[1][-0.5]
      \pnsemibox{pi1,?2}
      }}\right\}_{1 \leqslant j \leqslant n}
    \end{align*}
    \vspace{-\baselineskip}
    \caption{Non-empty box ($n > 0$)}
    \label{fig:eta-polybox}
  \end{subfigure}
  \caption{Actions of elementary paths on $\maltese$-cells and on a box in $\PolyPN$.}
  \label{fig:eta-actions-daimon}
\end{figure}
